\documentclass[reprint,longbibliography,superscriptaddress,amsmath,amssymb,aps,prx,twocolumn]{revtex4-2}

\usepackage{amsmath,amssymb,amsthm,mathrsfs,amsfonts,dsfont}
\usepackage{qcircuit}
\usepackage[dvips]{graphicx}
\usepackage{braket}
\usepackage{bm}
\usepackage{bbm}
\usepackage{enumerate}
\usepackage{color}
\usepackage{comment}
\usepackage{hyperref}
\hypersetup{
    colorlinks=true,
    linkcolor=blue,
    filecolor=blue,      
    urlcolor=blue,
    citecolor=blue,
    pdftitle={Efficient Simulation of Leakage Errors in Quantum Error Correcting Codes Using Tensor Network Methods},
    pdfpagemode=FullScreen,
}
\usepackage[normalem]{ulem}

\newtheorem{theorem}{Theorem}

\newtheorem{definition}{Definition}
\newtheorem{lemma}{Lemma}

\newtheorem{problem}{Problem}

\begin{document}
\title{Tensor Network Formulation of Dequantized Algorithms for Ground State Energy Estimation}

\author{Hidetaka Manabe}
\email{u687502j@ecs.osaka-u.ac.jp}
\affiliation{Graduate School of Engineering Science, The University of Osaka,
1-3 Machikaneyama, Toyonaka, Osaka 560-8531, Japan}

\author{Takanori Sugimoto}
\affiliation{Center for Quantum Information and Quantum Biology, The University of Osaka, 1-2 Machikaneyama, Toyonaka, Osaka 560-8531, Japan}
\affiliation{Computational Materials Science Research Team, RIKEN Center for Computational Science (R-CSS), Kobe, Hyogo 650-0047, Japan}
\affiliation{Advanced Science Research Center, Japan Atomic Energy Agency, Tokai, Ibaraki 319-1195, Japan}

\author{Keisuke Fujii}
\affiliation{Graduate School of Engineering Science, The University of Osaka,
1-3 Machikaneyama, Toyonaka, Osaka 560-8531, Japan}
\affiliation{Center for Quantum Information and Quantum Biology, The University of Osaka, 1-2 Machikaneyama, Toyonaka, Osaka 560-8531, Japan}
\affiliation{RIKEN Center for Quantum Computing (RQC), Hirosawa 2-1, Wako, Saitama 351-0198, Japan}

\begin{abstract}
Verifying quantum advantage for practical problems, particularly the ground state energy estimation (GSEE) problem, is one of the central challenges in quantum computing theory.
For that purpose, dequantization algorithms play a central role in providing a clear theoretical framework to separate the complexity of quantum and classical algorithms.
However, existing dequantized algorithms typically rely on sampling procedures, leading to prohibitively large computational overheads and hindering their practical implementation on classical computers.
In this work, we propose a tensor network-based dequantization framework for GSEE that eliminates the sampling process while preserving the asymptotic complexity of prior dequantized algorithms.
In our formulation, the overhead arising from sampling is replaced by the growth of the bond dimension required to represent Chebyshev vectors as tensor network states.
Consequently, physical structure, such as entanglement and locality, is naturally reflected in the computational cost.
By combining this approach with tensor network approximations, such as Matrix Product States (MPS), we construct a practical dequantization algorithm that is executable within realistic computational resources.
Numerical simulations demonstrate that our method can efficiently construct high-degree polynomials up to $d=10^4$ for Hamiltonians with up to $100$ qubits, explicitly revealing the crossover between classically tractable and quantum advantaged regimes.
These results indicate that tensor network-based dequantization provides a crucial tool toward the rigorous, quantitative verification of quantum advantage in realistic many-body systems.

\end{abstract}
\maketitle

\section{Introduction}
Quantum computers have attracted significant attention for their potential to perform information processing tasks that are intractable for classical computers~\cite{arute_Quantum_2019a, kim_Evidence_2023a}.
Among these tasks, ground state energy estimation (GSEE)~\cite{lin_HeisenbergLimited_2022,ding_Even_2023,wang_Quantum_2023,yoshioka_Hunting_2024,lee_Evaluating_2023} is a central problem in condensed matter physics and quantum chemistry, where exponential quantum speedup is anticipated.
Formally, given an $n$-qubit Hamiltonian of the form $H = \sum_{i=1}^{\mathrm{poly}(n)} H_i,$ the goal of GSEE is to estimate the ground state energy of $H$ to within the target precision $\epsilon$.
When the Hamiltonian can be accessed on a quantum computer and a guiding state with sufficient overlap with the true ground state is available, this task can be efficiently solved using algorithms based on quantum phase estimation (QPE) or quantum singular value transformation (QSVT) with eigenvalue filtering~\cite{martyn_Grand_2021, gilyen_Quantum_2019}.
Reflecting this broad applicability and importance, there has been extensive research on implementing these algorithms on an early fault-tolerant quantum computer with reduced resource requirements or shallower circuit depth~\cite{lin_Nearoptimal_2020,lin_HeisenbergLimited_2022,dong_GroundState_2022,ding_Even_2023,wang_Quantum_2023,ni_lowdepth_2023,kiss_Early_2025}.

However, the exponential speedup associated with the GSEE problem has a more nuanced structure than previously understood.
Gharibian and Le Gall~\cite{gharibian_Dequantizing_2023} showed that, given a guiding state sufficiently close to the ground state, GSEE for a $k$-local Hamiltonian is BQP-complete when the required energy precision $\epsilon$ scales as $1/\mathrm{poly}(n)$, whereas it becomes efficiently solvable on a classical computer when the precision is constant.
In their work, they employed an approach known as {\it dequantization}: the process of constructing a classical analogue of a quantum algorithm by relaxing certain assumptions on the problem~\cite{tang_quantuminspired_2019,tang_Quantum_2021,tang_Dequantizing_2022,gharibian_Dequantizing_2023,gall_Classical_2024,legall_Robust_2025,chia_Samplingbased_2020}.
Through dequantization, they revealed that the computational hardness hinges critically on the scaling of a single parameter, the precision $\epsilon$: different scalings separate classically tractable regimes from those requiring full quantum computational power.
Hence, the approach of dequantization serves as a powerful theoretical tool for probing the boundary between classical and quantum computation and for rigorously verifying quantum advantage.

Despite its theoretical importance, the existing dequantized algorithms are computationally impractical.
Many dequantization algorithms rely on Monte Carlo sampling to estimate inner products between high-dimensional vectors~\cite{legall_Robust_2025}, leading to enormous sampling overheads that render practical implementation on classical hardware infeasible.
For example, in the dequantized algorithm of Gharibian and Le Gall~\cite{gharibian_Dequantizing_2023} for the GSEE problem of a $k$-local Hamiltonian, the runtime scales as $O(2^{(8+k)/\epsilon})$.
Even for $\epsilon=0.1$, which is extremely coarse and far from practically meaningful, the resulting classical cost becomes prohibitively high.
As a result, prior studies on dequantization have largely remained at the level of computational complexity theory, without concrete implementations or numerical demonstrations.
This motivates the development of practically executable dequantization algorithms that preserve the theoretical insights while being computationally feasible on classical devices.

\begin{figure*}[t]
    \centering
    \includegraphics[width=0.9\linewidth]{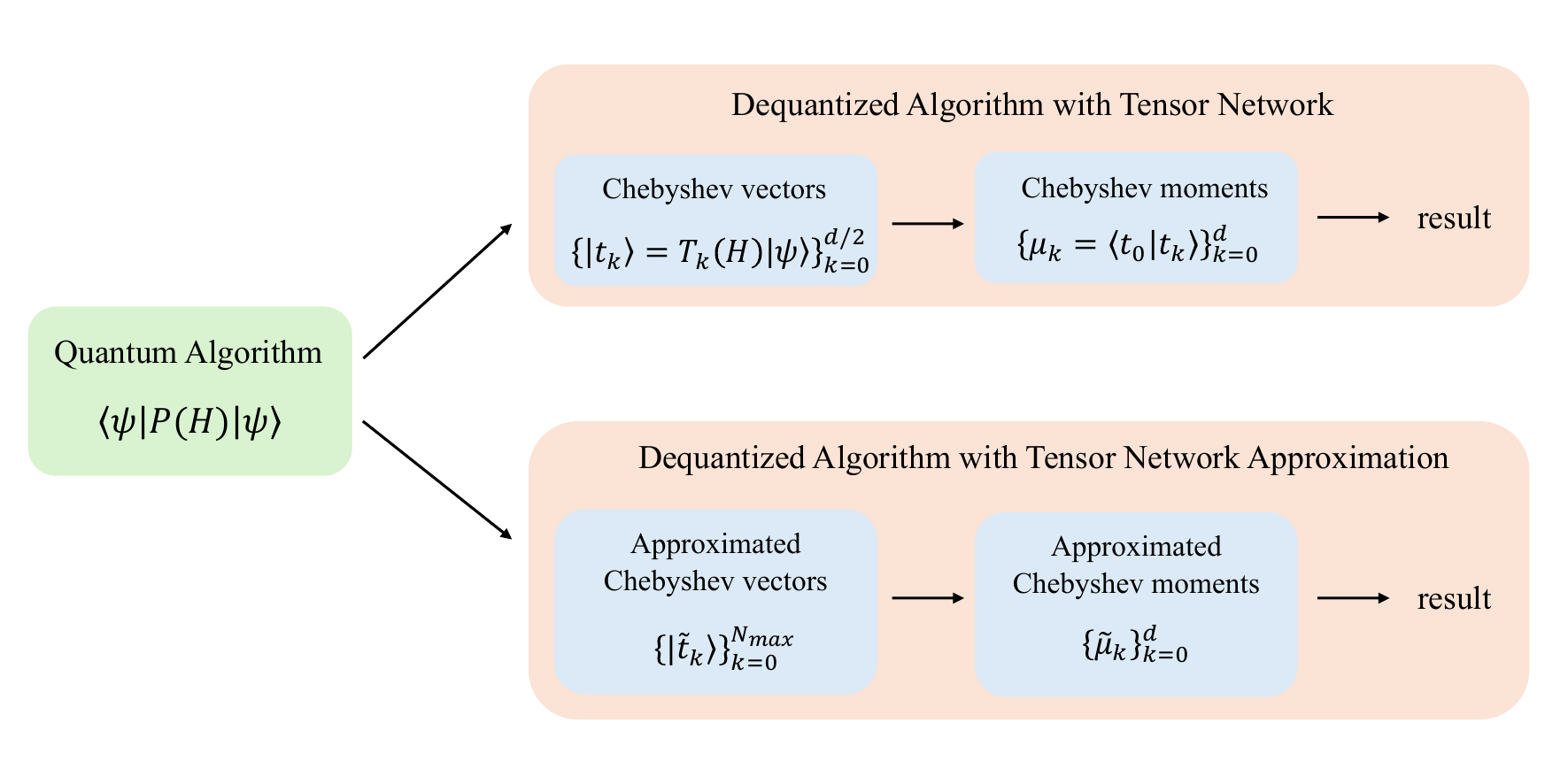}
    \caption{An overview of the dequantized algorithm with tensor networks for ground state energy estimation.
    $P$ denotes a polynomial of degree $d$, and $T_k$ denotes the $k$th Chebyshev polynomial.
    $\ket{\widetilde{t}_k}$ and $\widetilde{\mu}_k$ are approximations of the Chebyshev vector $\ket{t_k}$ and the moment $\mu_k$ respectively, and we assume $N_{\mathrm{max}}<d$.
    The formal definitions are provided later in the main text.
    }
    \label{fig:overview}
\end{figure*}

In this study, we introduce tensor network techniques~\cite{schollwock_densitymatrix_2011, orus_Practical_2014a, cirac_Matrix_2021, okunishi_Developments_2022, collura_Tensor_2024} to formulate a dequantized algorithm for solving the GSEE problem.
We classically simulate the QSVT algorithm with tensor networks, which share the same formulation as the Chebyshev Matrix Product States (ChebMPS) method~\cite{holzner_Chebyshev_2011, halimeh_Chebyshev_2015}, which was developed in condensed matter physics.
We reframe this technique within the rigorous framework of dequantization to quantitatively probe the boundary of quantum advantage.
An overview of our approach is shown in Fig.~\ref{fig:overview}.
We first expand the polynomial filter $P$ used in the quantum algorithm in the Chebyshev basis.
We then represent the Chebyshev vectors as tensor network states and compute the Chebyshev moments from these representations.
In this formulation, all quantities are evaluated deterministically through tensor network contraction.
As a result, we eliminate the sampling overhead arising from statistical error while preserving the same asymptotic complexity scaling.
Moreover, this formulation allows us to utilize the locality of the Hamiltonian to reduce the effective computational cost.

Furthermore, we introduce a practical variant of the dequantized algorithm that leverages tensor network approximation methods.
In this approach, we approximate the Chebyshev vectors using tensor network states and estimate the corresponding moments via linear prediction.
As long as these approximation errors remain within a tolerable range, the computational cost scales linearly---rather than exponentially---in $1/\epsilon$, yielding a dequantized algorithm that is executable in practice.
This formulation effectively recasts the source of computational hardness: the bottleneck shifts from the statistical variance of sampling to the entanglement growth of the Chebyshev vectors.

We implement the proposed method and benchmark it on one- and two-dimensional transverse-field Ising models (TFIM).
In particular, we use Matrix Product States (MPS)~\cite{perez-garcia_Matrix_2007} to represent Chebyshev vectors.
We demonstrate that our approach can efficiently construct high-degree polynomial filters up to $d=10^4$, which is fundamentally inaccessible to prior sampling-based approaches, and execute ground state energy estimation for $100$-qubit systems.
As a result, while the 1D TFIM can be dequantized with tensor networks to high precision, the 2D TFIM cannot be fully dequantized with the bond dimensions used in our experiments.
These results, however, explicitly visualize the crossover between classically tractable and quantum-advantaged regimes through a single precision parameter $\epsilon$, which captures the essence of the dequantization approach.
In this way, our framework can serve as a key methodology for the quantitative verification of quantum advantage in realistic many-body systems.

The paper is organized as follows.  
In Sec.~\ref{sec:preliminaries}, we review the tensor network methods and the GSEE
algorithm used in this work.  
In Sec.~\ref{sec:dequantize_exact}, we formulate the dequantized algorithm based on
exact tensor network contractions, and in Sec.~\ref{sec:dequantize_approx}, we extend
the formulation to include tensor network approximation methods.  
Section~\ref{sec:results} presents numerical results for one- and two-dimensional
transverse-field Ising models, demonstrating the practical utility of our approach.  
Finally, Sec.~\ref{sec:conclusion} summarizes our contributions.

\section{preliminaries}\label{sec:preliminaries}
In this section, we provide the tools required for developing a tensor network formulation of the dequantized algorithm for the GSEE problem.
We first introduce the basic concept of tensor networks, and then review the eigenvalue filtering techniques used in quantum and dequantized algorithms for GSEE.

\subsection{Tensor network}
A tensor network~\cite{orus_Practical_2014a, cirac_Matrix_2021, okunishi_Developments_2022} is a graphical formalism that represents a high-rank tensor as a contraction of many small tensors.
Tensor networks are widely used as a computational tool in fields such as condensed matter physics~\cite{white_Density_1992,nishino_Corner_1996,vidal_Efficient_2004,verstraete_Renormalization_2004,vidal_Entanglement_2007,schollwock_densitymatrix_2011}, data science~\cite{oseledets_TensorTrain_2011,stoudenmire_Supervised_2016,cichocki_Tensor_2016,han_Unsupervised_2018} and quantum computing~\cite{vidal_Efficient_2003,huggins_Quantum_2019,huang_Efficient_2021,collura_Tensor_2024} to avoid the so-called curse of dimensionality, where the memory cost increases exponentially with the number of dimensions.
We define tensor networks on an undirected graph.

\begin{definition}[Tensor network]
A tensor network is a triple $(G, \{T_v\}_{v \in V}, \{d_e\}_{e \in E})$ consisting of:
\begin{itemize}
    \item An undirected graph $G = (V, E)$.
    \item A positive integer $d_e\in \mathbb{N}$ associated with each edge $e \in E$, called the \emph{bond dimension}.
    \item For each vertex $v \in V$, a tensor
    \begin{equation}
        T_v \in \bigotimes_{e \in E_v} \mathbb{C}^{d_e},
    \end{equation}
    where $E_v:=\{e\in E: e\ni v\}$ denotes the set of edges incident to $v$.
\end{itemize}
\end{definition}
We distinguish between \emph{internal edges}, which connect two vertices (including loops), and \emph{dangling edges}, which are incident on only one vertex. 

\begin{definition}[Tensor network states]
A \emph{tensor network state} $TNS(G, \{T_v\}_{v \in V}, \{d_e\}_{e \in E})$ is a quantum state obtained by contracting all internal edges of the tensor network, i.e., summing over all indices associated with internal edges:
\begin{equation}
\ket{\mathrm{TNS}(G,\{T_v\},\{d_e\})}
=
\sum_{\{i_e\}_{e\in E_{\mathrm{int}}}}
\prod_{v\in V}
T_v\bigl(\{i_e\}_{e\in E_v}\bigr),
\end{equation}
where $E_\mathrm{int}$ denotes the set of internal edges.
\end{definition}

Throughout this study, we consider qubit systems and assume that $d_e = 2$ whenever $e$ is a dangling edge.

\begin{definition}[Inner product of tensor network states]
Let $(G, \{T_v\}_{v \in V}, \{d_e\}_{e \in E})$ and $(G, \{T'_v\}_{v \in V}, \{d'_e\}_{e \in E})$ be two tensor network states defined on the same graph $G = (V, E)$.
Suppose $d_e = d'_e$ if $e$ is a dangling edge.
Then, their inner product is defined by connecting the two tensor networks along the dangling edges and contracting:
\begin{align}
\braket{\mathrm{TNS}' | \mathrm{TNS}}
&=
\sum_{\{i_e\}_{e\in E}}
\sum_{\{j_e\}_{e\in E}}
\notag\\
&\quad\times
\prod_{v \in V}
\overline{T'_v(\{j_e\}_{e \in E_v})}
\,
T_v(\{i_e\}_{e \in E_v})
\notag\\
&\quad\times
\prod_{e \in E_{\mathrm{dang}}}
\delta_{i_e, j_e}.
\end{align}
Here, the overline denotes complex conjugation, $E_\mathrm{dang}$ denotes the set of dangling edges, and $\delta$ is a Kronecker delta.
\end{definition}

\begin{definition}[Efficiently contractible tensor network states]
We say that a family of tensor network states $\mathrm{TNS}(G, \{T_v\}, \{d_e\})$ that represent an $n$-qubit state is \emph{efficiently contractible} if the inner product between two tensor network states $\mathrm{TNS}(G, \{T_v\}, \{d_e\})$ and $\mathrm{TNS}(G, \{T'_v\}, \{d'_e\})$ can be computed in time polynomial in the number of qubits and the maximum bond dimension, i.e., in $\mathrm{poly}(n, \max_e d_e, \max_e d'_e)$ time.
\end{definition}

\begin{figure}[t]
    \centering
    \includegraphics[width=0.9\linewidth]{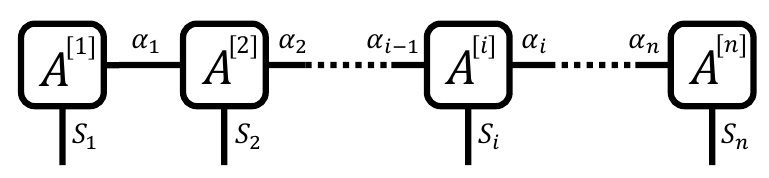}
    \caption{The diagram notation of MPS, a tensor network with a simple 1D structure. $\{s_i\}$ is a set of indices that corresponds to the dangling edges, and $\{\alpha_i\}$ is a set of internal indices.}
    \label{fig:MPS}
\end{figure}

A typical example of an efficiently contractible tensor network state is the \emph{matrix product states} (MPS).
An MPS is a tensor network defined on a 1D chain graph and can be written as:
\begin{equation}
\ket{\psi} = \sum_{\{s_i\}}\sum_{\{\alpha_i\}}\left[A_{\alpha_1}^{[1]s_1}A_{\alpha_1\alpha_2}^{[2]s_2}\cdots A_{\alpha_{n-1}}^{[n]s_n}\right]\ket{s_1s_2\cdots s_n}
\end{equation}
where $s_i$ denotes the physical index of a $i$th qubit and $\alpha_i$ denotes the internal index.
We illustrate the diagram notation of MPS in Fig.~\ref{fig:MPS}.

\begin{figure}[t]
    \centering
    \includegraphics[width=0.7\linewidth]{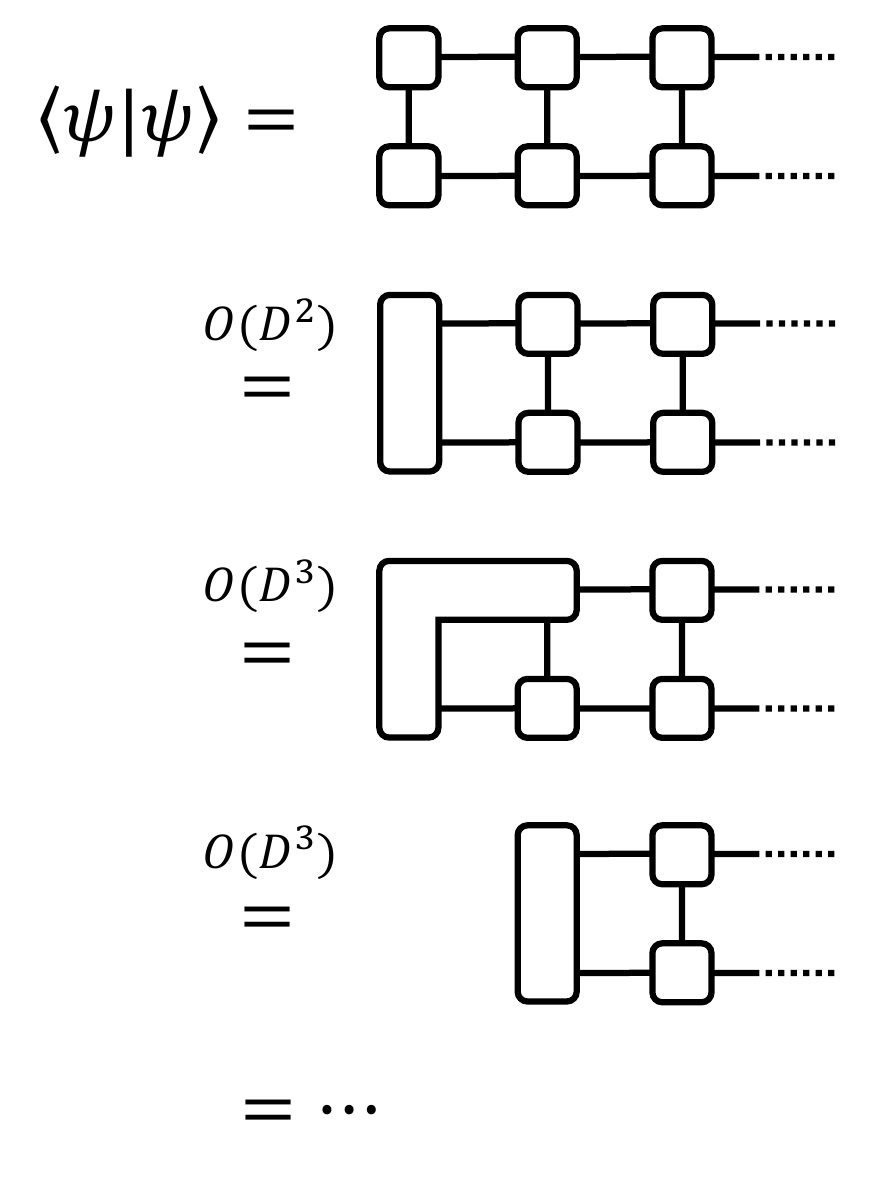}
    \caption{The process of calculating the inner product of two MPS. The double-layered MPS is contracted from the left-hand side.
    Each contraction costs $O(D^3)$, and this process is repeated $n$ times.}
    \label{fig:MPS_inner}
\end{figure}

Let $D = \max_e d_e$.
Then, the cost of computing the inner product scales as $\mathcal{O}(n D^3)$, indicating that MPS is efficiently contractible.
In Fig.~\ref{fig:MPS_inner}, we express the process of calculating the inner product of MPS in diagrammatic notation.

\subsection{Eigenvalue filtering}
In this section, we briefly review the eigenvalue filtering algorithm based on Quantum Singular Value Transformation (QSVT)~\cite{martyn_Grand_2021,gilyen_Quantum_2019}, which forms the foundation of the dequantized algorithm for GSEE.
\subsubsection{QSVT}
The QSVT algorithm applies a polynomial transformation to the singular values of a general matrix $A$.
In the context of GSEE, however, the matrix of interest is a non-negative Hermitian operator $H \geq 0$.
In this case, QSVT effectively performs quantum eigenvalue transformation (QEVT), in which a polynomial is applied directly to the eigenvalues of $H$.
Consequently, we adopt the QEVT framework throughout this work.
We also assume that the Hamiltonian is normalized: $\|H\|\leq 1$.

We assume access to $H$ through a block-encoding in a unitary operator $U$:
\begin{equation}
U = 
\begin{pmatrix}
H & * \\
* & *
\end{pmatrix}.
\end{equation}
Let $P\colon [-1, 1] \to \mathbb{R}$ be a real polynomial with degree $d$ satisfying:
\begin{itemize}
    \item P has parity $d \bmod 2$
    \item $|P(x)|\leq 1,\quad x\in[-1, 1]$.
\end{itemize}
The QEVT framework allows us to efficiently implement a unitary operator $U_P$ such that
\begin{equation}
U_P = 
\begin{pmatrix}
P(H) & * \\
* & *
\end{pmatrix}.
\end{equation}
Here, $P(H)$ is the matrix polynomial defined as
\begin{equation}
P(H) = \sum_{i} P(\lambda_i)\ket{\lambda_i}\bra{\lambda_i}
\end{equation}
where $\lambda_i$ is an eigenvalue of $H$ and $\ket{\lambda_i}$ is the corresponding eigenvector.

\subsubsection{Filtering function}
\begin{figure}[t]
    \centering
    \includegraphics[width=0.9\linewidth]{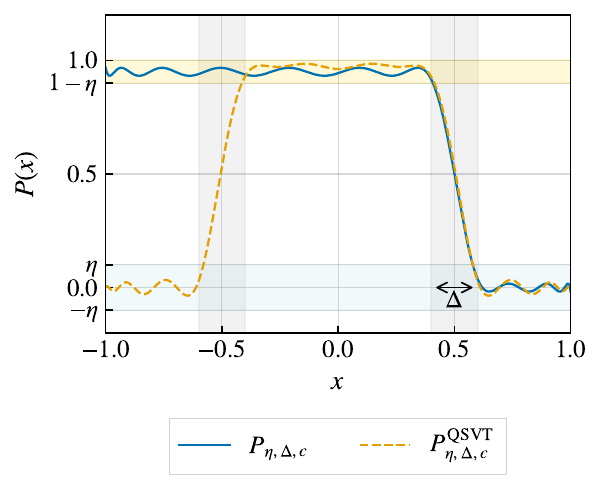}
    \caption{An example of polynomial approximation of the shifted sign function $P_{\eta,\Delta,c}$ and the rectangle function $P_{\eta,\Delta,c}^{\mathrm{QSVT}}$. Parameters are set as $\eta=0.1, \Delta=0.2, c=0.5$ and the degree is $20$.
    The polynomial approximation is obtained by CVXPY~\cite{diamond_CVXPY_2016}.}
    \label{fig:ET}
\end{figure}

The core of the QSVT-based GSEE algorithm lies in the design of the polynomial $P$.
By approximating the \emph{shifted sign function} with a polynomial and applying it as $P(H)$, one can filter out all eigenvalues of $H$ above a given threshold, enabling efficient ground state preparation and energy estimation.
Efficient polynomial approximations of the shifted sign function are well studied.

\begin{lemma}[Polynomial approximation of the shifted sign function~\cite{low_Hamiltonian_2017,gilyen_Quantum_2019}]
Let $\eta,\Delta\in(0,\frac{1}{2})$ and $c\in[-1,1]$. There exists a polynomial $P_{\eta,\Delta,c}(x)$ of degree $O\left(\frac{1}{\Delta}\log(1/\eta)\right)$ such that
\begin{equation}
    |P_{\eta,\Delta,c}(x)|\leq 1
\end{equation}
for all $x\in[-1,1]$ and 
\begin{equation}
\left\{
\begin{array}{ll}
P_{\eta,\Delta,c}(x) \in [1-\eta, 1] & \text{for } x \in \left[-1, c - \Delta/2 \right], \\
P_{\eta,\Delta,c}(x)  \in [-\eta, \eta] & \text{for } x \in \left[ c + \Delta/2, 1 \right].
\end{array}
\right.
\end{equation}
\end{lemma}

When implementing this transformation via QSVT, the polynomial $P$ must satisfy the parity constraint.
Therefore, instead of directly using a general shifted sign function, we consider the following rectangle function tailored for QSVT.

\begin{lemma}[Polynomial approximation of the rectangle function~\cite{low_Hamiltonian_2017,gilyen_Quantum_2019}]
Let $\eta,\Delta\in(0,\frac{1}{2})$ and $c\in[-1,1]$. There exists an even polynomial $P_{\eta,\Delta,c}^{\mathrm{QSVT}}(x)$ of degree $O\left(\frac{1}{\Delta}\log(1/\eta)\right)$ such that
\begin{equation}
    |P_{\eta,\Delta,c}^{\mathrm{QSVT}}(x)|\leq 1
\end{equation}
for all $x\in[-1,1]$ and 
\begin{equation}
\left\{
\begin{array}{ll}
P_{\eta,\Delta,c}^{\mathrm{QSVT}}(x) \in [1-\eta, 1] 
& \text{for } x \in [-c+\Delta/2,\; c-\Delta/2], \\[6pt]
P_{\eta,\Delta,c}^{\mathrm{QSVT}}(x) \in [-\eta, \eta] 
& \text{for } x \in [-1,\; -c-\Delta/2] \\
& \qquad\quad\,\,\cup\,[\,c+\Delta/2,\; 1\,].
\end{array}
\right.
\end{equation}
\end{lemma}

For visual intuition, we plot examples of polynomial approximations $P_{\eta,\Delta,c}(x)$ and $P_{\eta,\Delta,c}^{\mathrm{QSVT}}(x)$ in Fig.~\ref{fig:ET}.

\subsubsection{GSEE}
Now we define the central problem studied in this work: ground state energy estimation~\cite{lin_HeisenbergLimited_2022,ding_Even_2023,wang_Quantum_2023,yoshioka_Hunting_2024,lee_Evaluating_2023}.
\begin{problem}[GSEE]
Let $H$ be a Hamiltonian on $n$ qubits satisfying $H\geq 0$ and $\|H\|\leq 1$, i.e., $H$ is non-negative and normalized.
Let $\mathcal{E}(H)$ denote the ground state energy of $H$, namely, its smallest eigenvalue.
Given a precision parameter $\epsilon > 0$, the goal is to compute an estimate $\hat{E}$ such that
\begin{equation}
\left| \hat{E} - \mathcal{E}(H) \right| \leq \epsilon.
\end{equation}
\end{problem}

This problem can be efficiently solved on a quantum computer using QSVT with an appropriate filter function, provided that one has access to a block encoding of $H$ and a guiding state with sufficient overlap with the ground state.

\begin{lemma}[GSEE via QSVT~\cite{martyn_Grand_2021}]\label{lemma:GSEE_QSVT}
Suppose we are given a unitary block encoding of $H$.
Assume that a quantum state $\ket{\psi}$ can be prepared such that it has overlap $\chi := |\braket{\psi | \lambda_0}|$ with the ground state $\ket{\lambda_0}$ of $H$.
Then, for any $\delta > 0$, there exists a quantum algorithm that uses
\begin{equation}
\mathcal{O} \left( \frac{1}{\epsilon \chi^2} \log\left( \frac{1}{\chi} \right) \log\left( \frac{1}{\delta} \right) \right)
\end{equation}
queries to the block encoding of $H$ and outputs an $\epsilon$-approximation of $\mathcal{E}(H)$
with probability at least $1 - \delta$.
\end{lemma}
\begin{proof}
Given the initial state $\ket{0}\ket{\psi}$, applying the QSVT circuit for 
$P_{\eta,\Delta,c}^{\mathrm{QSVT}}(H)$ yields
\begin{equation}
U_{P}\bigl(\ket{0}\ket{\psi}\bigr)
=
\ket{0}\, P_{\eta,\Delta,c}^{\mathrm{QSVT}}(H)\ket{\psi}
\;+\;
\ket{1}\ket{\mathrm{garbage}}.
\end{equation}
Hence, the probability of measuring the outcome $0$ in the ancilla qubit is
\begin{equation}
p_0
=
\left\|
P_{\eta,\Delta,c}^{\mathrm{QSVT}}(H)\ket{\psi}
\right\|^{2}
=
\sum_{i} |c_i|^{2} 
\left|P_{\eta,\Delta,c}^{\mathrm{QSVT}}(\lambda_i)\right|^{2},
\label{eq:pc-definition}
\end{equation}
where $\ket{\psi}=\sum_i c_i\ket{\lambda_i}$.

We consider the following two cases.

\textbf{Case (1):} $\lambda_i \ge c+\Delta/2$ for all $i$.

In this regime, we have 
$\bigl|P_{\eta,\Delta,c}^{\mathrm{QSVT}}(\lambda_i)\bigr|\le \eta$ for all eigenvalues,
and hence
\begin{equation}
p_0
\le
\sum_i |c_i|^2\, \eta^2
= \eta^2.
\label{eq:case1}
\end{equation}

\textbf{Case (2):} $\lambda_0 \le c-\Delta/2$.

Let $\chi = |\braket{\lambda_0|\psi}|$ be the ground-state overlap.
Since 
$\bigl|P_{\eta,\Delta,c}^{\mathrm{QSVT}}(\lambda_0)\bigr| \ge 1-\eta$,
\begin{equation}
p_0
\ge 
|c_0|^2 (1-\eta)^2
=
\chi^2(1-\eta)^2.
\label{eq:case2}
\end{equation}

Therefore, by choosing the polynomial approximation error as $\eta=\chi/4$ for example, equations~\eqref{eq:case1} and \eqref{eq:case2} imply that the two probabilities satisfy
\begin{equation}
p_0\le \chi^2/16
\qquad\text{vs.}\qquad
p_0\ge \frac{9}{16}\chi^2
\end{equation}
By sampling the ancilla outcome, one can distinguish between the two cases.
Lemma 5 in Ref.~\cite{martyn_Grand_2021} implies that 
\begin{equation}
\mathcal{O}\!\left(
    \frac{1}{\chi^{2}}
    \log\frac{1}{\delta}
\right)
\label{eq:sample-complexity}
\end{equation}
samples suffice to determine the correct case with probability at least $1-\delta$.

We now apply this classifier to perform a binary search.
Assume an initial interval $[l,r]$ satisfies $l < \lambda_0 < r$.
At each iteration, define
\begin{equation}
c = \frac{l+r}{2},
\qquad
\Delta = \frac{r-l}{3},
\end{equation}
and perform the above decision procedure for the threshold~$c$.

\begin{itemize}
    \item If the ancilla measurement returns $0$ with high probability, then we are not in Case~(1),
    and thus
    \begin{equation}
    \lambda_0 < c+\frac{\Delta}{2}=\frac{l+2r}{3},
    \qquad
    r \leftarrow \frac{l+2r}{3}.
    \end{equation}

    \item Otherwise we are not in Case~(2), implying
    \begin{equation}
    \lambda_0 > c-\frac{\Delta}{2}=\frac{2l+r}{3},
    \qquad
    l \leftarrow \frac{2l+r}{3}.
    \end{equation}
\end{itemize}

In each iteration the interval shrinks by a constant factor $2/3$, and therefore
after
\begin{equation}
    \mathcal{O}\!\left(\log\frac{1}{\epsilon}\right)
\end{equation}
steps we obtain an interval $[l,r]$ of width at most~$2\epsilon$.
We output $\hat{E}=(l+r)/2$ as the final estimate.

Each iteration uses the QSVT circuit for $P_{\epsilon,\Delta,c}^{\mathrm{QSVT}}(H)$, whose degree is
\begin{equation}
O\!\left(\frac{1}{\Delta}\log\frac{1}{\eta}\right),
\end{equation}
and 
$\mathcal{O}\!\left(\frac{1}{\chi^{2}}\log\frac{1}{\delta}\right)$
measurements of the ancilla qubit.
Since $\Delta$ scales with the current interval size and decreases geometrically,
the total number of Hamiltonian queries is
\begin{equation}
\mathcal{O}\!\left(
    \frac{1}{\epsilon\,\chi^{2}}
    \log\frac{1}{\chi}
    \log\frac{1}{\delta}
\right).
\end{equation}
\end{proof}

We can confirm that this quantum algorithm remains efficient when $\epsilon=1/poly(n)$.

\subsubsection{dequantization}
We now consider {\it dequantizing} this quantum algorithm.
In this study, we define dequantization as a process of deriving a classical analogue of a quantum algorithm by relaxing several assumptions of the problem~\cite{gharibian_Dequantizing_2023, gall_Classical_2024,legall_Robust_2025,tang_Dequantizing_2022,chia_Samplingbased_2020,sakamoto_quantum_2025,tang_quantuminspired_2019,zhang_Dequantized_2024,tang_Quantum_2021,wu_Efficient_2024,wild_Classical_2023,shin_Dequantizing_2024,bakshi_Improved_2024}.
To dequantize, we impose some additional assumptions.
We consider a $k$-local Hamiltonian on $n$ qubits:
\begin{equation}
    H=\sum_{i=1}^{m} H_i,
\end{equation}
where each term $H_i$ acts non-trivially only on $k$ qubits.
We also assume that we have sample-and-query access to the guiding state $\ket{\psi}$:
\begin{enumerate}
    \item Query access: efficiently compute the amplitude $\braket{\bm{x}|\psi}$ for any bitstring $\bm{x} \in \{0,1\}^n$.
    \item Sample access: efficiently sample a bitstring $\bm{x} \in \{0,1\}^n$ with probability
    $|\braket{\bm{x}|\psi}|^2$.
\end{enumerate}

Under these additional assumptions, the GSEE problem can be solved with the following computational complexity:

\begin{lemma}[GSEE via dequantization~\cite{gharibian_Dequantizing_2023,gall_Classical_2024}]\label{lemma:GSEE_DEQ}
Let $H$ be a $k$-local Hamiltonian on $n$ qubits and assume that $\sum_{i=1}^{m}\|H_i\|=1$.
Suppose we are given sample-and-query access to a guiding state $\ket{\psi}$ with overlap $\chi$.
Then, there exists a classical algorithm that runs in time
\begin{equation}
    O^*\!\left(\frac{1}{\epsilon}\cdot 2^{(k+8)/\epsilon}\cdot\chi^{-8}\right)
\end{equation}
, and outputs an $\epsilon$-approximation of $\mathcal{E}(H)$ with probability at least $1 - 1/\exp(n)$.
\end{lemma}

\begin{proof}[Proof sketch]
We provide only a brief proof sketch and refer to prior work for the full details.
The basic idea is the same as in the QSVT-based approach: we aim to calculate
\begin{equation}
    \braket{\psi|P_{\eta,\Delta,c}(H)|\psi}
\end{equation}
and use this quantity to run a binary search to estimate the ground state energy.
In contrast to the quantum setting, we do not need to enforce positivity of $H$, nor impose any parity constraints on the polynomial.
Thus, we use $P_{\eta,\Delta,c}$ instead of $P_{\eta,\Delta,c}^{\mathrm{QSVT}}$.

As in the previous section, the polynomial satisfies
\begin{equation}
\lambda_i \ge c+\Delta/2 \quad \Rightarrow\quad
|\braket{\psi|P_{\eta,\Delta,c}(H)|\psi}| \leq \eta \label{eq:deq_case1}
\end{equation}
and
\begin{equation}
\lambda_0 \le c-\Delta/2 \quad \Rightarrow\quad
|\braket{\psi|P_{\eta,\Delta,c}(H)|\psi}| \geq \chi^2(1-\eta). \label{eq:deq_case2}
\end{equation}
Therefore, choosing $\eta = \chi^2/4$ allows us to distinguish the two cases by estimating the
expectation value to additive accuracy $O(\chi^2)$.

We expand the polynomial in monomials:
\begin{align}
    &P_{\eta,\Delta,c}(x)=\sum_{r=0}^{d} a_r x^r \\
    &\braket{\psi|P_{\eta,\Delta,c}(H)|\psi} = \sum_{r=0}^d a_r\braket{\psi|H^r|\psi}
\end{align}
To estimate each term $\braket{\psi|H^r|\psi}$, we write
\begin{equation}
    \braket{\psi|H^r|\psi} = \sum_{\bm{x}\in[m]^r}\braket{\psi|H_{x_1}\cdots H_{x_r}|\psi}
\end{equation}
where $[m]=\{1,2,\dotsc,m\}$. For each multi-index $\bm{x}$, define
\begin{equation}
    q(\bm{x})=\|H_{x_1}\|\cdots \|H_{x_r}\|.
\end{equation}
We introduce the random variable
\begin{equation}
    X=\frac{\braket{\psi|H_{x_1}\cdots H_{x_r}|\psi}}{q(\bm{x})},
\end{equation}
and average over $t$ independent samples, where $\bm{x}$ is sampled according to the probability $q(\bm{x})$.
By taking $t$ sufficiently large, we can get an accurate estimator of $\mathbb{E}[X]=\braket{\psi|H^r|\psi}$.

The main challenge is to compute $\braket{\psi|H_{x_1}\cdots H_{x_r}|\psi}$ in  $\mathrm{poly}(n)$ time rather than $\exp(n)$ time.
In fact, this quantity can be efficiently estimated, given that each $H_i$ is $k$-local (therefore $s$-sparse matrix with $s=2^k$) and we have sample-and-query access to $\ket\psi$.
The result of iterated matrix multiplication shows that one can estimate $\braket{\psi|H_{x_1}\cdots H_{x_r}|\psi}$ in time
\begin{equation}
    O^*\!\left(s^r\,{\epsilon'}^{-2}\,\log(1/\delta)\right),
\end{equation}
to additive error $\epsilon'\,\|H_{x_1}\|\cdots\|H_{x_r}\|$ with probability at least $1-\delta$.
Here the factor $\epsilon'^{-2}$ accounts for the sampling overhead of the estimation of the inner product, and $s^r$ accounts for the number of elements that are required for obtaining one element of $H_{x_1}\cdots H_{x_r}\ket\psi$.

Setting
\begin{equation}
  \epsilon' = O(\chi^2/4^r)
\quad
t = O(4^{2r}/\chi^{4})
\quad
\delta = O(1/8t),
\end{equation}
we obtain an estimator of $\langle\psi|H^r|\psi\rangle$ with additive error
$O(\chi^2/4^r)$.
The resulting time cost is
\begin{equation}
    O^*\!\left(t \cdot s^r\cdot 2^{4r}\cdot\chi^{-4}\log t\right)
    =
    O^*\!\left(s^r2^{8r}\chi^{-8} r\right).
\end{equation}
The coefficients of the polynomial bounded on $[-1, 1]$ satisfy~\cite{sherstov_Making_2012}
\begin{equation}
    |a_r| \le 4^r. \label{eq:coef_monomial}
\end{equation}
Moreover, the degree of the polynomial scales as $O^*(1/\epsilon)$. 
Putting everything together, the overall time complexity to estimate the quantity $\braket{\psi|P_{\eta,\Delta,c}(H)|\psi}$ in additive error $O(\chi^2)$ scales as
\begin{equation}
    O^*\!\left(\frac{1}{\epsilon}\cdot 2^{(k+8)/\epsilon}\cdot\chi^{-8}\right).
\end{equation}
\end{proof}

From the complexity, we see that the $1/\epsilon$ term appears in the exponent whose base is the constant ($2^{(k+8)}$).
Thus, when the target precision $\epsilon$ is a constant, the GSEE problem can, in principle, be solved efficiently on a classical computer. 

However, in practice, the constant overhead $2^{(k+8)/\epsilon}$ is prohibitively large, making it impossible to execute this algorithm in a realistic regime.
For example, if we set $k=2$, the term $2^{10/\epsilon}$ already implies that only rather coarse precision $\epsilon \gtrsim 0.1$ is feasible even when we have massive computing resources.
Since the Hamiltonian spectrum is normalized to the interval $[-1, 1]$, such a requirement on the precision makes the algorithm entirely impractical.
To the best of our knowledge, no work has implemented such sampling-based dequantized algorithms on a classical computer or evaluated their practical performance.

\section{Dequantized algorithm with Tensor Networks}\label{sec:dequantize_exact}
As discussed in the previous section, existing dequantized algorithms suffer from a large sampling overhead, making them practically infeasible to implement.
In this work, we propose an alternative formulation based on tensor networks that eliminates the need for sampling.
To do so, we need some additional assumptions that the guiding state and Hamiltonian are efficiently contractible.

\subsection{Additional Assumptions on the problem}

In all tensor network-based dequantized algorithms considered in this work, we assume that the guiding state is efficiently contractible.
In practice, this means that the guiding state is given in a tensor network representation such as an MPS.
Importantly, an efficiently contractible guiding state automatically satisfies the sample-and-query access assumption used in prior dequantization results.
To see this, recall that query access requires the ability to compute $|\braket{\psi|\bm{x}}|$ for any bitstring $\bm{x}$.
This quantity is a special case of contracting a tensor network with a product state and can therefore be computed efficiently whenever the guiding state is efficiently contractible.
Similarly, sampling from the distribution $p(\bm{x})=|\braket{\psi|\bm{x}}|^2/\|\psi\|^2$ requires the ability to evaluate marginals like $p(x_1)$, which are also efficiently computable under this assumption.
Thus, the efficiently contractible assumption implies the sample-and-query access model of Lemma~\ref{lemma:GSEE_DEQ} and can be viewed as a special case of the assumptions used in previous dequantization work.

We also define the class of Hamiltonians that can be handled by our algorithm.
While this class does not include all Hamiltonians, it does include physically important families such as $k$-local Hamiltonians.

\begin{definition}[Efficiently contractible Hamiltonian]
Let $\ket{\mathrm{TNS}(G, \{T_v\}, \{d_e\})}$ be an efficiently contractible tensor network state.
Consider a Hamiltonian $H$ such that the action of $H$ on $\ket{\mathrm{TNS}}$ yields another tensor network state on the same graph:
\begin{equation}
H\ket{\mathrm{TNS}(G, \{T_v\}, \{d_e\})}=\ket{\mathrm{TNS}(G, \{T'_v\}, \{d'_e\})}
\end{equation}
and satisfies
\begin{equation}
\max_{e \in E} d_e' = \mathrm{poly}\left( \max_{e \in E} d_e, |V| \right).
\end{equation}
We call such an operator $H$ an \emph{efficiently contractible Hamiltonian}.
\end{definition}

A typical example is a $k$-local Hamiltonian:
\begin{equation}
    H=\sum_{i=1}^m H_i
\end{equation}
where $k=O(1)$ and $m=\mathrm{poly}(n)$.
Each local term $H_i$ can be expressed as a sum of Pauli products:
\begin{equation}
    H_i = \sum_{j=1}^{4^k}\alpha_{ij}P_{ij}
\end{equation}
For a single Pauli term $P_{ij}$, the action $P_{ij} \ket{\mathrm{TNS}}$ does not increase any bond dimension $d_e$.
A sum of tensor network states with the same graph can be represented by a TNS:
\begin{align}
    &\ket{\mathrm{TNS}(G, \{T_v\}, \{d_e\})} = \\
     &\ket{\mathrm{TNS}(G, \{T^{(1)}_v\}, \{d^{(1)}_e\})} + \ket{\mathrm{TNS}(G, \{T^{(2)}_v\}, \{d^{(2)}_e\})}
\end{align}
with increased bond dimension:
\begin{equation}
    d_e = d^{(1)}_e + d^{(2)}_e.
\end{equation}
Therefore, after applying $H$ to the initial tensor network state, 
\begin{align}
    &\ket{\mathrm{TNS}(G, \{T'_v\}, \{d'_e\})} \notag \\
    &= H\ket{\mathrm{TNS}(G, \{T_v\}, \{d_e\})} \notag \\
    &=\sum_{i=1}^{m}\sum_{j=1}^{4^k}\alpha_{ij}P_{ij}\ket{\mathrm{TNS}(G, \{T_v\}, \{d_e\})},
\end{align}
the resulting bond dimensions $d_e'$ satisfy
\begin{equation}
\max_e d_e' = m\cdot4^k \cdot \max_e d_e.
\end{equation}

Even if the Hamiltonian is not $k$-local, a similar argument holds for a Pauli Hamiltonian:
\begin{equation}
    H = \sum_{i=1}^m P_i,
\end{equation}
where $P_i$ is an $n$-qubit Pauli operator and $m=\mathrm{poly}(n)$. 
Because the action $P_i\ket{\mathrm{TNS}}$ does not increase the bond dimension, the resulting bond dimension $d'_e$ satisfies
\begin{equation}
\max_e d_e' = m\cdot\max_e d_e.
\end{equation}

\subsection{Chebyshev Polynomial Expansion}
The core of the GSEE algorithm based on eigenvalue filtering lies in evaluating the quantity $\langle \psi | P(H) | \psi \rangle$. In this work, we expand the polynomial filter $P(H)$ in terms of Chebyshev polynomials and compute the above expression using recurrence relations~\cite{holzner_Chebyshev_2011, halimeh_Chebyshev_2015,weisse_Kernel_2006}.

Let $T_k(x)$ denote the Chebyshev polynomial of the first kind, defined by the recurrence relations:
\begin{align}
T_0(x) &= 1, \\
T_1(x) &= x, \\
T_k(x) &= 2x T_{k-1}(x) - T_{k-2}(x) \quad \text{for } k \geq 2.
\end{align}

We expand the polynomial $P(x)$ in a Chebyshev series of degree $d$:
\begin{equation}
P(x) = \frac{a_0}{2} + \sum_{k=1}^{d} a_k T_k(x),
\end{equation}
where $a_k$ are the expansion coefficients. Then, the quantity of interest becomes
\begin{align}
\braket{\psi|P(H)|\psi} &= \frac{a_0}{2}+\sum_{k=1}^{d} a_k \braket{\psi | T_k(H) | \psi} \notag \\
&=\frac{a_0}{2} + \sum_{k=1}^d a_k\mu_k,  \label{eq:quantity_moment}
\end{align}
where we define the {\it Chebyshev moments} as $\mu_k = \braket{\psi|T_k(H)|\psi}$.

Let $\ket{t_k} := T_k(H) \ket{\psi}$ be the $k$th {\it Chebyshev vector}, which can be recursively computed as:
\begin{align}
\ket{t_0} &= \ket{\psi}, \label{eq:recurrence0}\\
\ket{t_1} &= H \ket{\psi}, \label{eq:recurrence1}\\
\ket{t_k} &= 2 H \ket{t_{k-1}} - \ket{t_{k-2}}. \label{eq:recurrence2}
\end{align}
Then, Chebyshev moments are calculated from the Chebyshev vectors as follows~\cite{weisse_Kernel_2006,holzner_Chebyshev_2011,halimeh_Chebyshev_2015}:
\begin{align}
    \mu_{2k} &= 2\braket{t_k|t_k}-\mu_0, \label{eq:moments_even}\\
    \mu_{2k+1} &= 2\braket{t_{k+1}|t_k} - \mu_1.\label{eq:moments_odd}
\end{align}
Therefore, to calculate \eqref{eq:quantity_moment} of order $2d$, Chebyshev vectors up to order $d$ are required.

The use of Chebyshev polynomials, rather than monomials, is motivated by their superior numerical stability.  
For any integer $k$, the Chebyshev vector $T_k(H)\ket{\psi}$ admits the eigenbasis expansion
\begin{equation}
    T_k(H)\ket{\psi}
    =
    \sum_i c_i\, T_k(\lambda_i)\ket{\lambda_i},
\end{equation}
which oscillates as $T_k(\lambda_i)=\cos(k\theta_i)$ when we write $\lambda_i = \cos\theta_i$.  
Hence the Chebyshev vectors remain $O(1)$ in norm for all $k$: they do not exhibit exponential growth or decay.

Moreover, if a function $f(x)$ satisfying $|f(x)|\le 1$ on $[-1,1]$ is expanded into Chebyshev polynomials,
\begin{equation}
    f(x)=\frac{a_0}{2} + \sum_{k\ge 1} a_k\, T_k(x),
\end{equation}
then the coefficients satisfy the uniform bounds
\begin{equation}
    |a_0|\le 1,\qquad 
    |a_k|\le \frac{4}{\pi}\quad (k\ge 1). \label{eq:cheb_coef_bound}
\end{equation}
This is in sharp contrast to the monomial expansion~\eqref{eq:coef_monomial}, whose coefficients grow exponentially in $k$.  
In the latter case, the estimation error of $\braket{\psi|H^k|\psi}$ is amplified by a factor of $O(4^k)$, making reliable estimation rapidly infeasible.

These stability properties, the $O(1)$ norm of Chebyshev vectors and Chebyshev coefficients, are essential for implementing dequantized algorithms in practical numerical settings, where issues such as loss of accuracy and floating-point roundoff can easily arise.
They are particularly crucial when approximation techniques are applied, as discussed later.

\subsection{Complexity Analysis via Tensor Networks}

We analyze the computational cost of evaluating $\langle \psi | P(H) | \psi \rangle$ using tensor network contraction.

\begin{lemma}[Computing moments with tensor network states]\label{lemma:contract_moments}
Assume that the Hamiltonian $H$ and the guiding state $\ket{\psi}=\ket{\mathrm{TNS}(G, \{T_v\}, \{d_e\})}$ are efficiently contractible. Then the moment $\mu_d=\braket{\psi | T_d(H) | \psi}$ can be computed in time
\begin{equation}
\mathrm{poly}((D_H)^{d/2}, D, n),
\end{equation}
where $D_H$ is the maximum bond dimension growth induced by applying $H$, $d$ is the degree of the polynomial, and $D$ is the maximum bond dimension of the initial guiding state.
\end{lemma}
\begin{proof}
We represent the Chebyshev vectors as tensor network states:
\begin{equation}
    \ket{t_k}=T_k(H)\ket{\mathrm{TNS}(G, \{T_v\}, \{d_e\})}.
\end{equation}
Let $D^{(k)}$ be the maximum bond dimension of the $k$th Chebyshev vector $\ket{t_k}$.
$D_H$ is the maximum growth in bond dimension when $H$ acts on a tensor network state:
\begin{equation}
    \max_e d_e' \leq D_H \max_e d_e,
\end{equation}
where $d_e'$ is the bond dimension of $H\ket\psi$.
Then from the recurrence relation \eqref{eq:recurrence0}--\eqref{eq:recurrence2},
\begin{align}
D^{(0)} &= D, \\
D^{(1)} &= D_H D, \\
D^{(k)} &= D_H D^{(k-1)} + D^{(k-2)}.
\end{align}
To leading order, this implies
\begin{equation}
D^{(k)} = \mathcal{O}((D_H)^k D).
\end{equation}
Because each Chebyshev vector can be represented by a tensor network state whose bond dimension is polynomial in $D$ and $(D_H)^k$, their inner products can also be calculated in time polynomial in $D, (D_H)^k$ and $n$.
By using \eqref{eq:moments_even} and \eqref{eq:moments_odd}, we can efficiently compute $\braket{\psi | T_d(H) | \psi}$ using Chebyshev vectors up to $O(d/2)$ degree, which proves the claim.
\end{proof}

\begin{theorem}\label{th:1}
Assume that the Hamiltonian $H$ and the guiding state $\ket{\psi}=\ket{\mathrm{TNS}(G, \{T_v\}, \{d_e\})}$ are efficiently contractible,
and the guiding state has an overlap $\chi = |\braket{\psi | \lambda_0}|$ with the ground state of $H$.
Let $D=\max_e d_e$ and $D_H$ be the maximum bond dimension growth induced by applying $H$.
Then, there exists a classical algorithm that computes an $\epsilon$-approximation of the ground state energy $\mathcal{E}(H)$ in time
\begin{equation}
    \mathrm{poly}((D_H)^{1/(2\epsilon)}, D, n) 
\end{equation}
\end{theorem}

\begin{proof}
We consider an approximate shifted sign function $P_{\eta,\epsilon,c}(x)$ with degree $d = \mathcal{O}(1/\epsilon \cdot \log(1/\eta))$, and expand it in the Chebyshev basis:
\begin{equation}
    P_{\eta,\epsilon,c}(x) = \frac{a_0}{2} + \sum_{k= 1}^{d}a_kT_k(x).
\end{equation}
From equation \eqref{eq:cheb_coef_bound}, the coefficients $\{a_k\}$ are uniformly bounded by a constant.

We can compute the quantity $\braket{\psi|P_{\eta,\epsilon,c}(H)|\psi}$ using the Chebyshev moments:
\begin{equation}
    \braket{\psi|P_{\eta,\epsilon,c}(H)|\psi} = \frac{a_0}{2} + \sum_{k= 1}^{d} a_k\mu_k
\end{equation}
By Lemma~\ref{lemma:contract_moments}, these moments can be computed in time $\mathrm{poly}(D_H^{d/2}, D, n)$.
The procedure for estimating the ground state energy from $\braket{\psi|P_{\eta,\epsilon,c}(H)|\psi}$ is the same as in Lemma \ref{lemma:GSEE_QSVT} and \ref{lemma:GSEE_DEQ}.
\end{proof}

As an example, we estimate the computational cost when the guiding state is given as an MPS and the Hamiltonian is a Pauli Hamiltonian.
In this case, the increase in bond dimension from applying the Hamiltonian is $D_H = m$, where $m$ is the number of Pauli terms.
Taking into account the cost of inner product computation for MPS, the total cost becomes
\begin{equation}
O^*\left(D^3 m^{1.5/\epsilon}\right).
\end{equation}
Compared to the sampling-based approach whose cost scales as $O^*\!\left(\frac{1}{\epsilon}\cdot 2^{(k+8)/\epsilon}\cdot\chi^{-8}\right)$, this formulation removes the polynomial dependence on the overlap $\chi$ and eliminates heavy factors such as $2^{(k+8)/\epsilon}$ that arise from sampling overhead.
On the other hand, the dependence on $m$, the number of terms in the local Hamiltonian, becomes explicit.
Since $m = O(n)$ in many physically relevant cases, the tensor network-based approach becomes more expensive in regimes where $n$ is large.

However, it is important to note that this is a worst-case estimate.
For instance, in the case of a nearest-neighbor 2-local Hamiltonian, we have $D_H = 3$, leading to an overall complexity of
\begin{equation}
    O^*\left(D^3 3^{1.5/\epsilon}\right),
\end{equation}
which is significantly smaller.
Furthermore, as discussed in the next section, the use of tensor network approximation techniques allows one to heuristically suppress the bond dimension, thereby reducing the computational cost in practice.

\section{Practical dequantized algorithm with tensor network approximation}\label{sec:dequantize_approx}
In the previous sections, we analyzed the computational complexity of the GSEE problem using tensor networks. 
Under similar assumptions to those in prior works based on sample-and-query access, the GSEE problem can be solved with polynomial time complexity in $n$ when the target precision is constant.
Moreover, by exploiting the locality of the Hamiltonian, our method achieves significantly lower computational cost than sampling-based approaches.
However, in both cases, the required bond dimension grows exponentially with the polynomial degree $d$, which limits the practicality of the algorithm~\cite{cifuentes_Quantum_2024,montanaro_Quantum_2024}.
As such, the current dequantized algorithm remains within the scope of computational complexity theory.

To address this limitation, we propose to use tensor network approximation in this study.
There exists a wide range of approximation techniques for tensor networks, and in many physically relevant systems, the bond dimension can be significantly reduced.
In our case, since the Chebyshev vectors are represented as tensor networks, the core challenge becomes how efficiently these vectors can be approximated using tensor network ansatz.
Furthermore, by extrapolating the Chebyshev moments with a linear prediction, we can effectively access much higher polynomial degrees.

\subsection{Approximated Chebyshev vectors}
We consider representing Chebyshev vectors $\ket{t_k}$ as an approximate tensor network state $\ket{\widetilde{t_k}}$, such as an MPS.
We define the accuracy of the approximated state $\Delta^{(k)}$ as the difference from the actual state:
\begin{equation}
    \Delta^{(k)} = \|\ket{t_k} - \ket{\widetilde{t_k}}\|.
\end{equation}
Then, the following property holds.

\begin{theorem}\label{th:2}
Assume the Hamiltonian $H$ and the guiding state $\ket{\psi}$ satisfy the conditions of Theorem~\ref{th:1}.
Let $d=O^*(1/\epsilon)$ be the degree required to construct $\braket{t_0|P_{\eta,\epsilon,c}(H)|t_0}$.
Consider the approximated tensor network states $\ket{\widetilde{t_k}}$ with maximum bond dimension $D_1$.
If the accuracy of these approximated states satisfies
\begin{equation}
    \sum_{k=1}^{d}\Delta^{(k)} \leq \frac{\pi}{32}\chi^2,  \label{eq:error_upper_bound}
\end{equation}
then we can compute an $\epsilon$-approximation of the ground state energy $\mathcal{E}(H)$ in time
\begin{equation}
    \mathrm{poly}(1/\epsilon, D_1, n).
\end{equation}
\end{theorem}
This theorem claims that, if we can obtain the approximated Chebyshev vectors with high quality, we can solve the GSEE problem in polynomial time in the degree $d$.
In this case, the exponential runtime overhead associated with $d$ has been removed, resulting in a practically feasible dequantized algorithm.

\begin{proof}
We estimate $\braket{t_0|P_{\eta,\epsilon,c}(H)|t_0}$ using approximate tensor network states:
\begin{equation}
    \braket{\psi|P_{\eta,\epsilon,c}(H)|\psi}\simeq \frac{a_0}{2}+\sum_{k=1}^{d}a_k\braket{t_0|\widetilde{t_k}}.
\end{equation}
Its error can be upper bounded by:
\begin{align}
    &|\braket{\psi|P_{\eta,\epsilon,c}(H)|\psi} - \frac{a_0}{2} - \sum_{k=1}^{d}a_k\braket{t_0|\widetilde{t_k}}| \notag\\
    &\leq \sum_{k=1}^d |a_k||\braket{t_0|t_k} - \braket{t_0|\widetilde{t_k}}| \notag \\
    &\leq \frac{4}{\pi} \cdot \sum_{k=1}^d \Delta^{(k)}
\end{align}
If the error of the tensor network approximation is upper bounded by \eqref{eq:error_upper_bound}, then the estimation error of $\braket{\psi|P_{\eta,\epsilon,c}(H)|\psi}$ is upper bounded by:
\begin{equation}
    |\braket{\psi|P_{\eta,\epsilon,c}(H)|\psi} - \frac{a_0}{2} - \sum_{k=1}^{d}a_k\braket{t_0|\widetilde{t_k}}| \leq \frac{\chi^2}{8}.
\end{equation}
From Eqs \eqref{eq:deq_case1} and \eqref{eq:deq_case2}, we can successfully distinguish two cases and execute a binary search to estimate the ground state energy. 
\end{proof}

For later discussions, let $\widetilde{\mu}_i$ be the approximated moments:
\begin{align}
    \widetilde{\mu}_{2k} &= 2\braket{\widetilde{t}_k|\widetilde{t}_k}-\widetilde{\mu}_0 \label{eq:appro_moments_even}\\
    \widetilde{\mu}_{2k+1} &= 2\braket{\widetilde{t}_{k+1}|\widetilde{t}_k} - \widetilde{\mu}_1\label{eq:appro_moments_odd}
\end{align}
We also define the moment approximation error $\Delta_m^{(k)}$ as:
\begin{equation}
    \Delta_{m}^{(k)}=|\mu_k - \widetilde{\mu}_k|.
\end{equation}

\begin{lemma}\label{lemma:moments}
Assume that the Hamiltonian $H$ and the guiding state $\ket{\psi}$ satisfy the conditions of Theorem~\ref{th:1}.
Let $d=O^*(1/\epsilon)$ be the polynomial degree required to construct $\braket{t_0|P_{\eta,\epsilon,c}(H)|t_0}$.
If the approximated moments $\{\widetilde{\mu}_k\}_{k=0}^d$ satisfy
\begin{equation}
    \sum_{k=1}^{d}\Delta_m^{(k)} \leq \frac{\pi}{32}\chi^2,  \label{eq:error_upper_bound_moment}
\end{equation}
then we can compute an $\epsilon$-approximation of the ground state energy $\mathcal{E}(H)$ in time
\begin{equation}
    \mathrm{poly}(d).
\end{equation}
\end{lemma}
\begin{proof}
The estimation error of $\braket{t_0|P_{\eta,\epsilon,c}(H)|t_0}$ can be bounded as follows:
\begin{align}
    &|\braket{\psi|P_{\eta,\epsilon,c}(H)|\psi} - \frac{a_0}{2} - \sum_{k=1}^{d}a_k\widetilde{\mu}_k| \notag \\
    &\leq \sum_{k=1}^d |a_k||\mu_k - \widetilde{\mu}_k| \notag \\
    &\leq \frac{4}{\pi} \sum_{k=1}^d \Delta_m^{(k)} \label{eq:moment_error_quantity}
\end{align}
\end{proof}

\subsection{Approximation method}
To apply the above theorem, it is necessary to perform an efficient tensor network approximation of the Chebyshev vectors.
In this work, we focus in particular on Matrix Product States (MPS), which are widely used to represent one-dimensional quantum many-body systems.

\subsubsection{Matrix Product States (MPS)}
MPS is efficiently contractible tensor network states typically used for 1D quantum systems.
We assume that the guiding state $\ket{t_0}$ is represented as an MPS. The approximate Chebyshev vectors are generated via the recurrence relations:
\begin{align}
    &\ket{\widetilde{t}_1'} = H\ket{t_0} \label{eq:MPS_t1}\\
    &\ket{\widetilde{t}_k'} = 2H\ket{\widetilde{t}_{k-1}} - \ket{\widetilde{t}_{k-2}}. \label{eq:MPS_tk}
\end{align}
Applications of $H$ increase the bond dimension of the MPS.
We impose a maximum bond dimension $\chi_{\mathrm{mps}}$ and truncate if the bond dimension of the network exceeds this threshold:
\begin{equation}
    \ket{\widetilde{t}_k} = \mathrm{truncate}(\ket{\widetilde{t}_k'}). \label{eq:MPS_appro}
\end{equation}
In this study, we use a standard MPS compression algorithm based on the canonical form and bond truncation with singular value decomposition (SVD)~\cite{perez-garcia_Matrix_2007,schollwock_densitymatrix_2011}.
Higher-quality approximations can also be obtained using iterative methods such as two-site DMRG~\cite{white_Density_1992}, which we did not employ due to their runtime overhead and relatively small improvements.

In practice, the truncation error at step $k$ is usually quantified by the {\it cosine error}:
\begin{equation}
    \Delta_c^{(k)} = \left|1 - \frac{\braket{\widetilde{t}_k|\widetilde{t}_k'}}{\|\ket{\widetilde{t}_k}\|\|\ket{\widetilde{t}_k'}\|}\right|.
\end{equation}
Note that the cosine error is a local fitting error for each step and differs from the global error $\Delta^{(k)}$ used in Theorem~\ref{th:2}.
Nevertheless, we will see in the section on numerical experiments that this cosine error serves as a good indicator of the actual approximation error.

The truncation error at step $k$ is also quantified by  \emph{truncation error} $\Delta_t^{(k)}$, defined as:
\begin{equation}
    \Delta_t^{(k)} = \|\ket{\widetilde{t}_k'} - \ket{\widetilde{t}_k}\|. \label{eq:truncation_error}
\end{equation}
This quantity is also efficiently calculated within the MPS framework.
The local error does not grow exponentially due to the property of the Chebyshev recursion, and we can obtain a uniform upper bound for local truncation errors to guarantee that the dequantization algorithm works.

\begin{theorem}\label{th:local_bound_simple}
Assume the conditions of Theorem~\ref{th:2} hold.
Let $\delta$ be a uniform upper bound on the local truncation error defined in Eq.~\eqref{eq:truncation_error}, i.e., $\Delta_t^{(k)} \leq \delta$ for all $k$.
Then, the global accuracy condition (Eq.~\eqref{eq:error_upper_bound}) is satisfied if:
\begin{equation}
    \delta \leq \frac{3 \pi \chi^2}{16 d^3} = O(\chi^2 d^{-3}).
\end{equation}
\end{theorem}

\begin{proof}
Let $\ket{e_k} = \ket{t_k} - \ket{\widetilde{t}_k}$ be the global error vector at step $k$.
The local error introduced at step $k$ is $\ket{\eta_k} = \ket{\widetilde{t}_k'} - \ket{\widetilde{t}_k}$, with norm $\|\ket{\eta_k}\| = \Delta_t^{(k)} \leq \delta$.
The error propagates according to the linear recurrence:
\begin{equation}
    \ket{e_k} = 2H\ket{e_{k-1}} - \ket{e_{k-2}} + \ket{\eta_k}.
\end{equation}
The local truncation error $\ket{\eta_j}$ introduced at step $j$ accumulates at a later step $k$ ($k \ge j$).
This contribution is given by:
\begin{equation}
    U_{k-j}(H) \ket{\eta_j},
\end{equation}
where $U_k(x)$ denotes the Chebyshev polynomial of the second kind of degree $k$.
A key property of these polynomials is that their maximum absolute value on the interval $[-1, 1]$ grows linearly with their degree:
\begin{equation}
    \max_{x \in [-1, 1]} |U_n(x)| = n+1.
\end{equation}
Since the eigenvalues of $H$ lie within $[-1, 1]$, the operator norm is bounded by $\|U_{k-j}(H)\| \leq k-j+1$.
Consequently, the norm of the local error $\ket{\eta_j}$ is amplified by a factor of at most $(k-j+1)$ when it propagates to step $k$.

Assuming the worst-case accumulation, the norm of the global error at step $k$ is bounded by:
\begin{equation}
    \|\ket{e_k}\| \leq \sum_{j=1}^{k} (k-j+1) \|\ket{\eta_j}\| \leq \delta \sum_{m=1}^{k} m \approx \frac{k^2}{2} \delta.
\end{equation}
Summing these global errors up to degree $d$ as required by Theorem~\ref{th:2}:
\begin{equation}
    \sum_{k=1}^{d} \|\ket{e_k}\| \lesssim \sum_{k=1}^{d} \frac{k^2}{2} \delta \approx \frac{d^3}{6} \delta.
\end{equation}
Requiring this total error to be bounded by $\frac{\pi}{32}\chi^2$, we obtain the condition $\frac{d^3}{6} \delta \leq \frac{\pi}{32}\chi^2$, which yields the stated bound.
\end{proof}

This theorem provides a sufficient condition for the approximated dequantized algorithm to execute with theoretical guarantees.
We emphasize, however, that this bound is likely conservative, as it is derived from a worst-case error analysis.
Specifically, the $O(d^{-3})$ dependence already imposes a requirement of $\Delta_t \sim 10^{-6}$ even for moderate degrees such as $d \sim 100$.
Given that the cosine error $\Delta_c$, used as a local fitting error in our numerical simulations, scales quadratically with the truncation error ($\Delta_c = O(\Delta_t^2)$), strictly satisfying this theoretical bound would correspond to a high-accuracy regime (e.g., $\Delta_c \sim 10^{-12}$).
In practice, our numerical results indicate that the algorithm can perform accurately even when the errors exceed this threshold, as we will discuss later.

\subsubsection{Linear Prediction} \label{sec:linear_prediction}
In our approximation strategy, we represent the Chebyshev vectors as tensor network states.
Due to entanglement growth, accurately representing these vectors requires sufficiently large bond dimensions.
However, computing tens of thousands of approximate tensor network states with such a large bond dimension is computationally demanding.

As shown in Lemma~\ref{lemma:moments}, only the approximate moments ${\widetilde{\mu}_k}$ are required for the ground state energy estimation, rather than the full tensor network states.
To this end, we adopt a linear prediction (LP) scheme~\cite{wolf_Chebyshev_2014,ganahl_Chebyshev_2014} to extrapolate the moment sequence $\{\mu_k\}$ beyond the range directly accessible to MPS.

Recall that the initial state is written in the eigenbasis of $H$:
\begin{equation}
    \ket{\psi} = \sum_i c_i\ket{\lambda_i}
\end{equation}
and define $\theta_i=\arccos{(\lambda_i)}$.
Each moment can be written as
\begin{align}
    \mu_k 
    &= \langle\psi|T_k(H)|\psi\rangle \notag \\
    &= \sum_i |c_i|^2 \cos(k\theta_i) 
    \label{eq:mu_oscillatory}
\end{align}
showing that the sequence is a finite linear combination of oscillatory modes.
If the initial state $\ket{\psi}$ has non-zero overlap with only a few eigenstates (i.e., only a small number of coefficients $|c_i|^2$ are dominant), then the moment sequence \eqref{eq:mu_oscillatory} is governed by only a few frequencies $\{\theta_i\}$.  
In this case, the effective dimensionality of the signal is low, and the sequence is well approximated by a low-order autoregressive model.  

Suppose we calculate the Chebyshev vectors up to degree $N_\mathrm{max}$ and have a moment sequence $\{\mu_k\}_{k=0}^{2N_\mathrm{max}}$.
We model the tail of the moment sequence using the relation
\begin{equation}
    \mu_n \simeq 
    -\sum_{j=1}^{n_{\mathrm{fit}}} a_j \mu_{n-j},
    \label{eq:ar_model}
\end{equation}
where $\{a_j\}$ are real coefficients.  
These coefficients are obtained by a least-squares fit of Eq.~\eqref{eq:ar_model} over the final $n\in\{2N_{\mathrm{max}}-n_{\mathrm{fit}},\dots,2N_{\mathrm{max}}-1\}$ moments.
Once the coefficients are determined, the moments for $n > 2N_{\mathrm{max}}$ are generated recursively using Eq.~\eqref{eq:ar_model}.  

In this way, the linear prediction procedure allows us to obtain $\{\mu_n\}$ for $n \gg 2N_{\mathrm{max}}$ without further MPS calculations.
Using these extrapolated moments, we can then approximately evaluate the higher-degree polynomial for $\langle\psi|P(H)|\psi\rangle$.
Note that this is a heuristic approach, and its accuracy depends on the initial state and the system we consider.
A theoretical characterization of when linear prediction yields reliable extrapolations in this context remains an important topic for future work.

\section{Results}\label{sec:results}
In this section, we implement the proposed method and evaluate its performance through numerical experiments on benchmark models.
The tensor network simulations are implemented using the quimb~\cite{gray_quimb_2018} library, and all computations are executed on a machine with an AMD EPYC 7532 32-core processor and an NVIDIA A100 40-GB GPU.

\subsection{Problem Settings}
\subsubsection{Target Model}
In this work, we consider the one-dimensional and two-dimensional transverse-field Ising models (1D TFIM and 2D TFIM) as benchmark models.
The 1D TFIM is defined on an open chain:
\begin{equation}
H_{\mathrm{1D}}
= -J \sum_{i=1}^{L-1} X_i X_{i+1}
    - h \sum_{i=1}^{L} Z_i ,
\end{equation}
and the 2D TFIM is defined on an open lattice with $L\times L$ sites:
\begin{equation}
H_{\mathrm{2D}}
= -J \sum_{\langle i,j\rangle} X_i X_j
    - h \sum_{i} Z_i .
\end{equation}

Depending on the ratio $h/J$, the system exhibits a quantum phase transition ($h/J=1.0$ for the 1D TFIM and $h/J\simeq 3.044$ for the 2D TFIM).
Because of its simplicity and rich physics, the TFIM is widely used as a benchmark model for testing numerical methods based on tensor networks.
In this work, we normalize $\|H\|\leq 1$ so that the eigenvalue filter can be applied.

\subsubsection{Chebyshev Polynomial Approximation}
\begin{figure*}[t]
    \centering
    \includegraphics[width=0.9\linewidth]{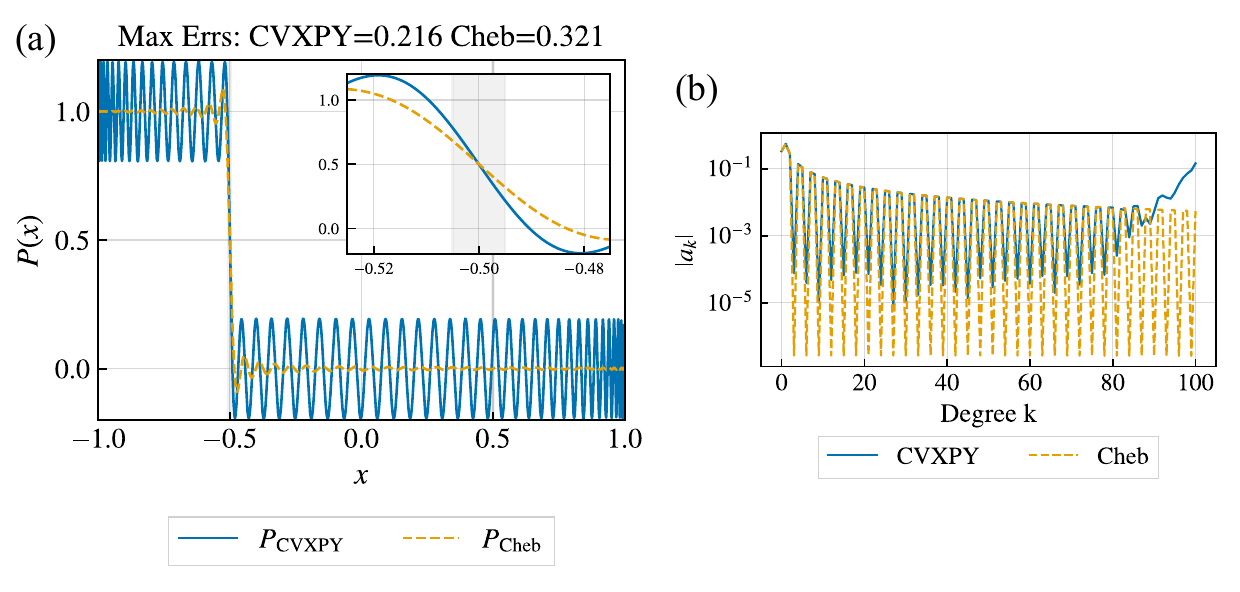}
    \caption{The comparison between two polynomial approximations. (a) The polynomials that approximate the shifted sign function are plotted. $P_{\mathrm{CVXPY}}$ is the one obtained by CVXPY and $P_{\mathrm{Cheb}}$ is obtained by Chebyshev truncation. The discontinuity point is set to $c=-0.5$. The inset is the zoom-in version around $x=c$. (b) The Chebyshev coefficients $|a_k|$ are plotted as a function of $k$.}
    \label{fig:figchebopt}
\end{figure*}

Several methods have been proposed for constructing Chebyshev polynomial approximations to the shifted sign function, including those based on numerical optimization~\cite{dong_GroundState_2022,kane_Nearly_2024}.
However, optimization-based approaches become prohibitively slow and unstable, especially when the degree exceeds a few thousand.

In this work, we instead generate the required polynomials from a Chebyshev expansion of the corresponding error function.
Fig.~\ref{fig:figchebopt}(a) compares two polynomial approximations of the shifted sign function: 
one obtained using CVXPY and the other obtained by truncating the Chebyshev expansion of the error function.  
The gap parameter is set to $\Delta = 0.01$ and the polynomial degree is set to $d=100$, and no parity constraint is imposed.
The maximum approximation error of the CVXPY polynomial is lower than that of the Chebyshev truncation: $0.216$ versus $0.321$.
This can be confirmed by inspecting the inset of Fig.~\ref{fig:figchebopt}(a): the CVXPY polynomial can reproduce the discontinuity better.
However, the CVXPY-generated polynomial exhibits strong oscillations outside the gap region, whereas the Chebyshev truncation remains stable except for the Gibbs oscillations near the discontinuity.

Fig.~\ref{fig:figchebopt}(b) plots the magnitude of the coefficients $|a_k|$ for both polynomials.
For the Chebyshev truncation, the coefficients decay exponentially with $k$, while the CVXPY polynomial exhibits larger coefficients at higher orders.  
This indicates that the CVXPY solution reduces the approximation error by actively employing higher-degree terms.  
Nevertheless, as seen from Eq.~\eqref{eq:moment_error_quantity}, the error in the quantity of interest involves the product of $|a_k|$ and $\Delta_m^{(k)}$.
Since $\Delta_m^{(k)}$ grows with $k$, a polynomial with exponentially decaying coefficients, such as that obtained by Chebyshev truncation, is far more desirable for classical simulation.

Although Chebyshev truncation is not optimal in the max-norm, its error differs from the optimal one by at most a constant factor.
Furthermore, the Chebyshev truncation approach remains numerically stable and computationally efficient even for degrees up to tens of thousands.
For these reasons, we adopt Chebyshev truncation as our method for constructing the polynomial approximations.

In principle, more sophisticated optimization methods may yield polynomials with higher accuracy and improved stability.
Such optimized constructions could be beneficial not only for classical simulations but also for practical implementations of QSVT, where the degree directly determines the circuit depth.  
Developing such optimized and robust polynomial constructions is an interesting direction for future work.

\subsubsection{Numerical Experiment Settings}

We adopt an MPS-based implementation of the dequantization algorithm for our numerical simulations.
To run the algorithm, we require an initial MPS with a sufficiently large overlap with the true ground state of the model.
We prepare the initial states by running DMRG using TeNPy~\cite{hauschild_Efficient_2018a}, with the bond dimension restricted to $\chi_{\mathrm{init}}$.

We then compute the Chebyshev vectors according to Eqs.~\eqref{eq:MPS_t1},~\eqref{eq:MPS_tk} and \eqref{eq:MPS_appro}, using an MPS approximation whose bond dimension is truncated to $\chi_{\mathrm{mps}}$ at every step. 
After obtaining $N_\mathrm{max}$ Chebyshev vectors, we compute the $2N_\mathrm{max}$ Chebyshev moments using Eqs.~\eqref{eq:appro_moments_even} and \eqref{eq:appro_moments_odd}.
We then apply linear prediction to extrapolate the moment sequence if needed.

Given access to the moments, we can evaluate expectation values $\braket{\psi|P(H)|\psi}$ for polynomial $P$. 
This enables a binary-search procedure for ground state energy estimation: by repeatedly shrinking the gap of a shifted sign function, we can exponentially refine the estimate of the ground energy.

For simplicity, in this work we fix the gap parameter as $\Delta = 1/d$, where $d$ is the maximal moment order used.  
Empirically, we have found that this choice suffices to construct a high-accuracy polynomial approximation via Chebyshev truncation.
For a fixed gap, we first construct a family of polynomials $\{P_{\eta, \Delta, x}\}_x$, where $x$ ranges from $-1$ to $1$.
We then evaluate the cumulative function
\begin{equation}
    C(x) := \frac{a_0^{(x)}}{2} + \sum_{k=1}^d a_k^{(x)}\, \mu_k,
\end{equation}
where $a_k^{(x)}$ denotes the Chebyshev coefficients of $P_{\eta, \Delta, x}$.
We take the value of $x$ for which $C(x)$ is closest to the threshold $\chi^{2}/2$, and return the interval $[x-\Delta/2, x+\Delta/2]$ as the resulting estimate with error $\Delta$.
The solution obtained in this way is reachable by the original binary-search procedure.

\subsection{Numerical Results}

\begin{figure*}[t]
    \centering
    \includegraphics[width=\linewidth]{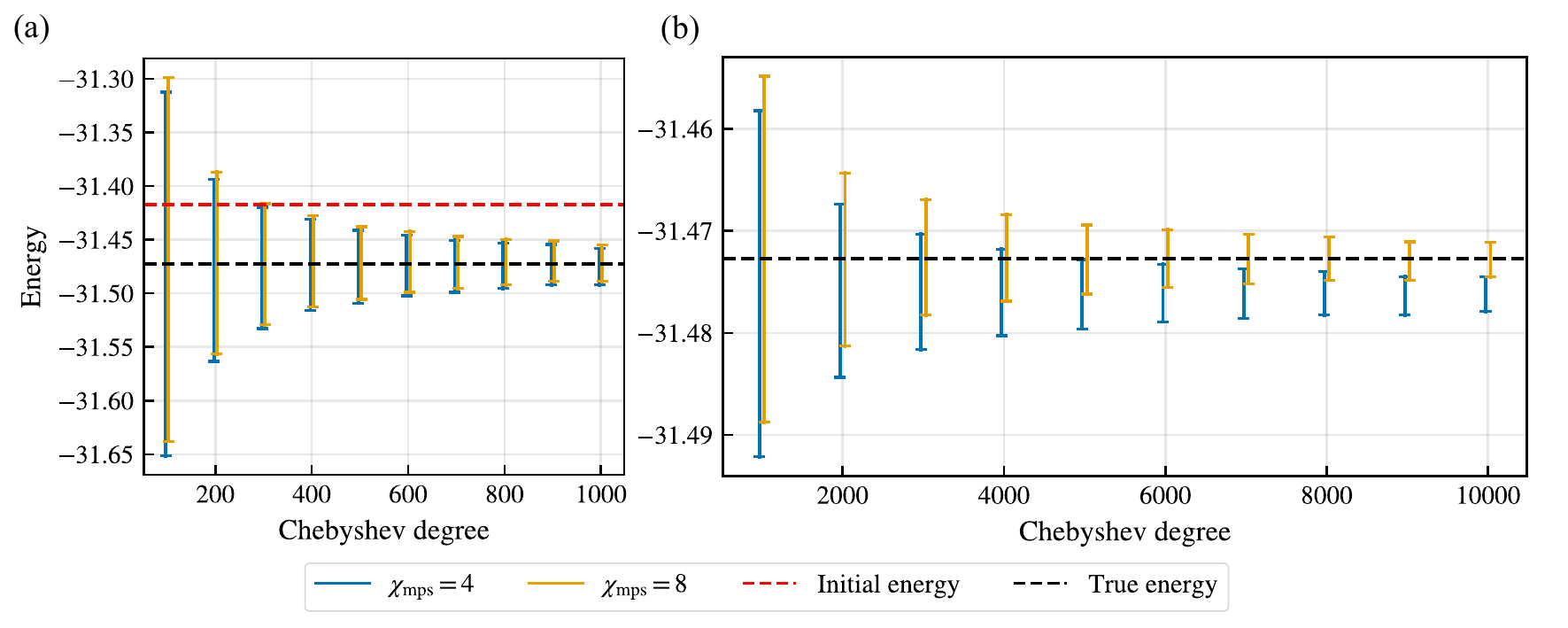}
    \caption{The estimated energy is plotted as a function of Chebyshev degree (a) from 100 to 1000 and (b) from 1000 to 10000 for the 1D TFIM with $L=25, J=1.0, h=1.0$.
    Each color represents a different maximum bond dimension $\chi_{\mathrm{mps}}$ used for the Chebyshev vector.
    The red dashed line represents the energy of the initial state, and the black dashed line represents the exact ground state energy. The initial energy is prepared with $\chi_{\mathrm{init}}=2$.}
    \label{fig:fig1}
\end{figure*}

To obtain a deeper understanding of the proposed algorithm, we first perform simulations in a regime where full state-vector simulation is still feasible.
Fig.~\ref{fig:fig1} shows the results of the ground-state energy estimation for the 1D TFIM. We consider a system of size $L=25$ at the critical point $J=1.0, h=1.0$. The initial state is prepared using DMRG with bond dimension $\chi_{\mathrm{init}}=2$, and each Chebyshev vector is approximated as an MPS with maximum bond dimension $\chi_{\mathrm{mps}}=4$ or $8$.
As shown in the figure, the dequantized algorithm yields a ground-state energy estimate that is more accurate than the energy of the initial state, indicating that our algorithm is working correctly.
Moreover, increasing the Chebyshev degree reduces the gap of the approximated threshold function, enabling higher-precision energy estimation. 
The degrees we reach, $d=10^4$, are far beyond what is achievable in the exact dequantized algorithm, including those that use Monte Carlo sampling.
This demonstrates that the approximated tensor network approach brings the algorithm into a practical regime while preserving the essence of the dequantized framework.

\begin{figure*}[!t]
    \centering
    \includegraphics[width=0.85\linewidth]{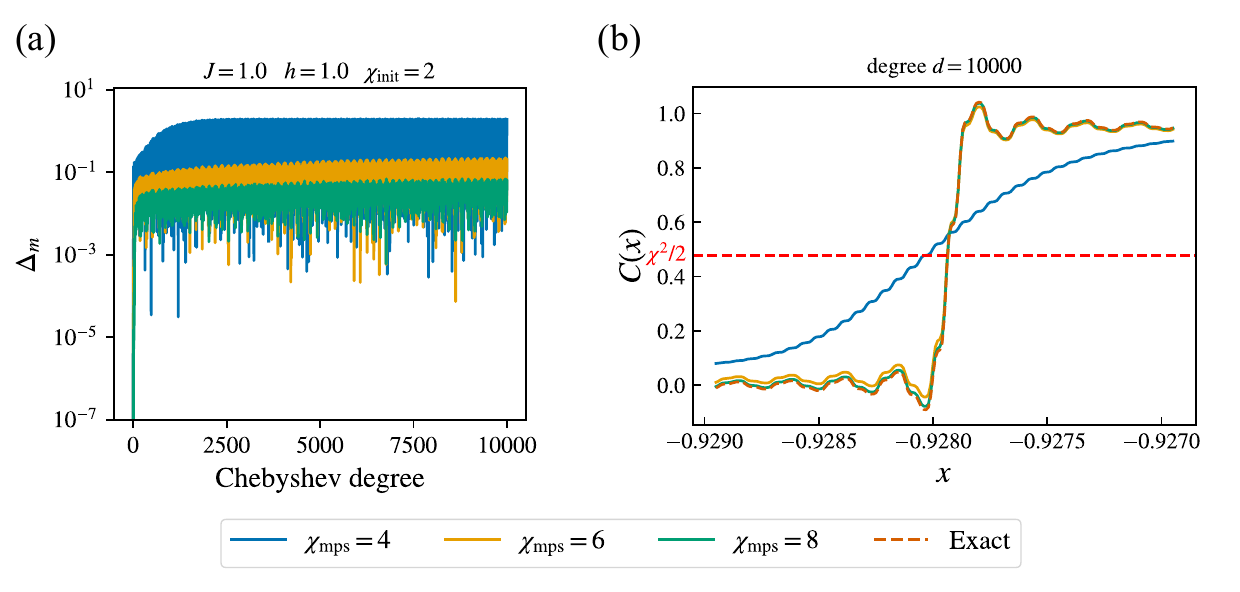}
    \caption{(a) The moment error $\Delta_m^{(k)}$ is plotted every $10$ steps as a function of Chebyshev degree $k$ for the 1D TFIM model.
    Each color represents a different maximum bond dimension of MPS. Parameters are set to $J=1.0, h=1.0, \chi_{\mathrm{init}}=2$.
    (b) The cumulative function $C(x)$ using $d$ moments is plotted every $10$ steps as a function of $x$.
    The dashed line shows the result with the exact moments obtained from the state-vector simulation.
    The red horizontal line is the threshold $\chi^2/2$ to estimate the ground state energy.}
    \label{fig:fig1Dacc}
\end{figure*}

We also plot the moment error $\Delta_m^{(k)}$ for various maximum bond dimensions $\chi_{\mathrm{mps}}$ in Fig.~\ref{fig:fig1Dacc}(a).
As expected, the moment error decreases as we increase the bond dimension and represent the Chebyshev vector more faithfully. 
Fig.~\ref{fig:fig1Dacc}(b) shows the cumulative function $C(x)$ constructed from the sequence of moments up to order $10^4$.
We observe that a larger bond dimension yields a more accurate reconstruction of $C(x)$, which in turn leads to a more reliable energy estimate.
These results indicate that increasing the bond dimension systematically improves the quality of the results, just as in conventional tensor network methods in condensed matter physics and data science.

\begin{figure*}[t]
    \centering
    \includegraphics[width=0.8\linewidth]{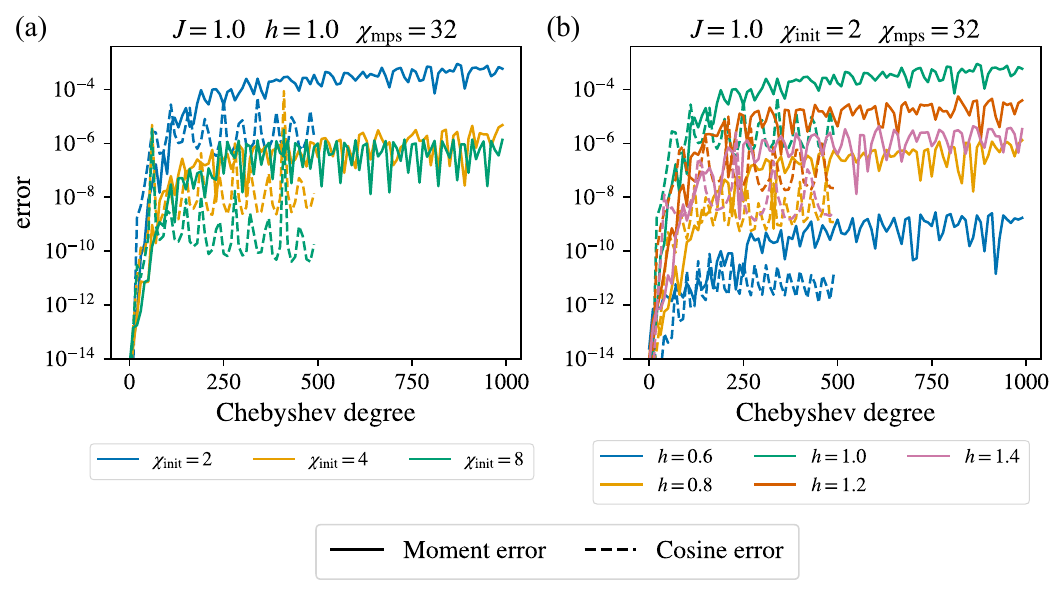}
    \caption{The cosine error at each step and the moment error for the 1D TFIM are plotted as functions of Chebyshev degree. (a) Different colors represent the bond dimension $\chi_{\mathrm{init}}$ for the initial state preparation. (b) Different colors represent the different values of transverse field $h$.}
    \label{fig:fig2}
\end{figure*}

Fig.~\ref{fig:fig2} presents the moment error and cosine error for the 1D TFIM under various parameter choices.
The dashed curves represent the cosine error obtained at each MPS truncation step, while the solid curves show the deviation from the exact Chebyshev moments computed using state-vector simulation.
Fig.~\ref{fig:fig2}(a) shows the effect of changing the initial-state bond dimension $\chi_{\mathrm{init}}$.
Interestingly, a larger $\chi_{\mathrm{init}}$ leads to a more accurate representation of the Chebyshev vectors, even when we use the same bond dimension $\chi_{\mathrm{mps}}$.
A good initial state is close to the ground state, which exhibits low entanglement entropy.
Therefore, it is reasonable to expect that the application of the Hamiltonian to generate the Chebyshev vectors will keep the resulting states near the low-entanglement manifold.
In the QSVT-based GSEE algorithm, a better initial state overlap $\chi$ leads to a more efficient runtime.
Our results suggest, however, that the improvement of the initial state also lowers the cost of the corresponding dequantized algorithm, thereby potentially reducing the quantum advantage.

Fig.~\ref{fig:fig2}(b) shows how the error varies with the transverse field $h$. 
The hardest instance occurs at $h=1.0$, while the error decreases as $h$ deviates from $1.0$.
These observations are consistent with the intuition that, in the 1D TFIM, $h=1.0$ corresponds to the critical point, where the initial state contains a broader set of eigenstates, making the corresponding Chebyshev vectors inherently more difficult to approximate.

Moreover, Fig.~\ref{fig:fig2} demonstrates that the behavior of the cosine error exhibits the same trend as the actual error in the Chebyshev moments.
This implies that the local fitting error strongly reflects the global error.
Despite the moment error being accessible only when an exact state-vector simulation is tractable, the cosine error can always be computed efficiently from the MPS. 
This confirms that the precision of the Chebyshev vectors can be reliably estimated from the cosine error.

\begin{figure*}[t]
    \centering
    \includegraphics[width=\linewidth]{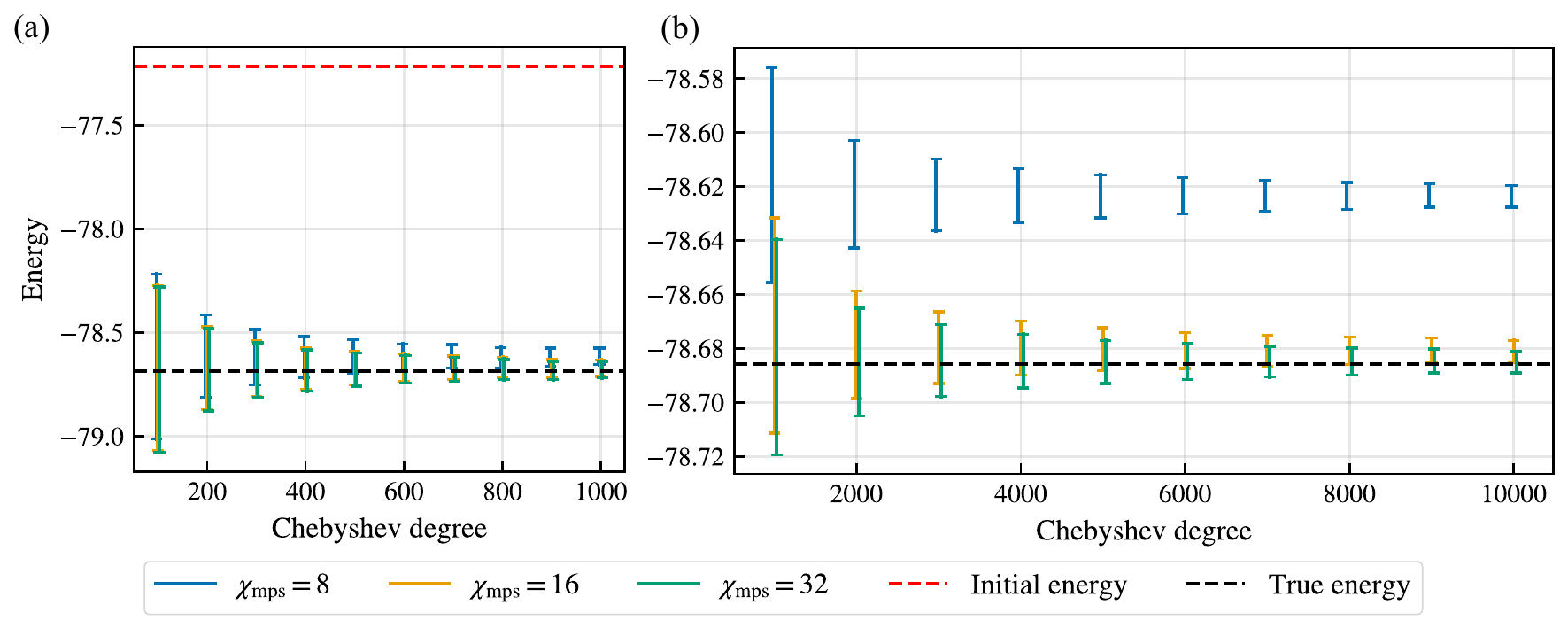}
    \caption{The estimated energy is plotted as a function of Chebyshev degree for the 2D TFIM with $L=5, J=1.0, h=3.0$.
    Each color represents a different maximum bond dimension $\chi_{\mathrm{mps}}$ for the Chebyshev vectors.
    The red dashed line represents the energy of the initial state, and the black dashed line represents the exact ground state energy. The initial energy is prepared with $\chi_{\mathrm{init}}=2$.}
    \label{fig:fig2_2}
\end{figure*}

\begin{figure*}[t]
    \centering
    \includegraphics[width=0.8\linewidth]{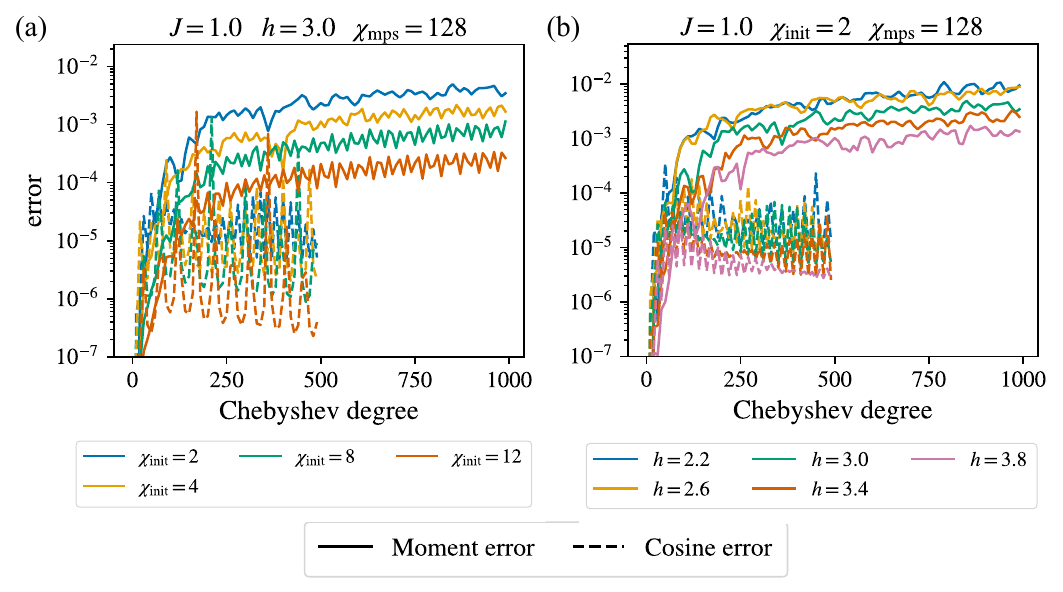}
    \caption{The cosine error for each step and the exact moment error for the 2D TFIM are plotted every $10$ steps as functions of Chebyshev degree. (a) Different colors represent the bond dimension for the initial state preparation. (b) Different colors represent the different transverse field $h$.}
    \label{fig:fig1_2}
\end{figure*}

We also perform simulations for the 2D TFIM and compare the results with the exact state-vector simulation. Fig.~\ref{fig:fig2_2} shows the results of the ground state energy estimation and Fig.~\ref{fig:fig1_2} shows the error in the Chebyshev moments for 2D TFIM.
Both figures exhibit trends similar to the 1D case: the dequantized algorithm yields energy estimates that are more accurate than that of the initial state, and increasing either the bond dimension or the quality of the initial state leads to more accurate approximations of the Chebyshev vectors.
The difficulty of the approximation also depends on the model parameters.
Note that the hardest point of $h$ deviates from the actual critical point because of finite size effects.
However, the overall accuracy is significantly worse than in the 1D case.
This degradation arises from the intrinsic difficulty of approximating a two-dimensional model using a one-dimensional tensor network ansatz.
In the 2D TFIM, the entanglement entropy of the ground state obeys an area law $S\sim \alpha L$.
Therefore, one can expect that the Chebyshev vectors also obey the same scaling, meaning that they can no longer be efficiently represented by MPS.
These results indicate that substantially larger bond dimensions are required to achieve high-precision approximations in two dimensions as the system size increases.

Finally, we present results for a $100$-qubit system, which is a typical target size where state-vector simulation is infeasible and a quantum computer would be required.
The MPS simulation for this size becomes computationally costly if we set the bond dimension $\chi_{\mathrm{mps}}$ large.
Therefore, we combine the Chebyshev MPS simulation with the linear prediction technique described in Sec.~\ref{sec:linear_prediction}.

\begin{figure*}[t]
    \centering
    \includegraphics[width=\linewidth]{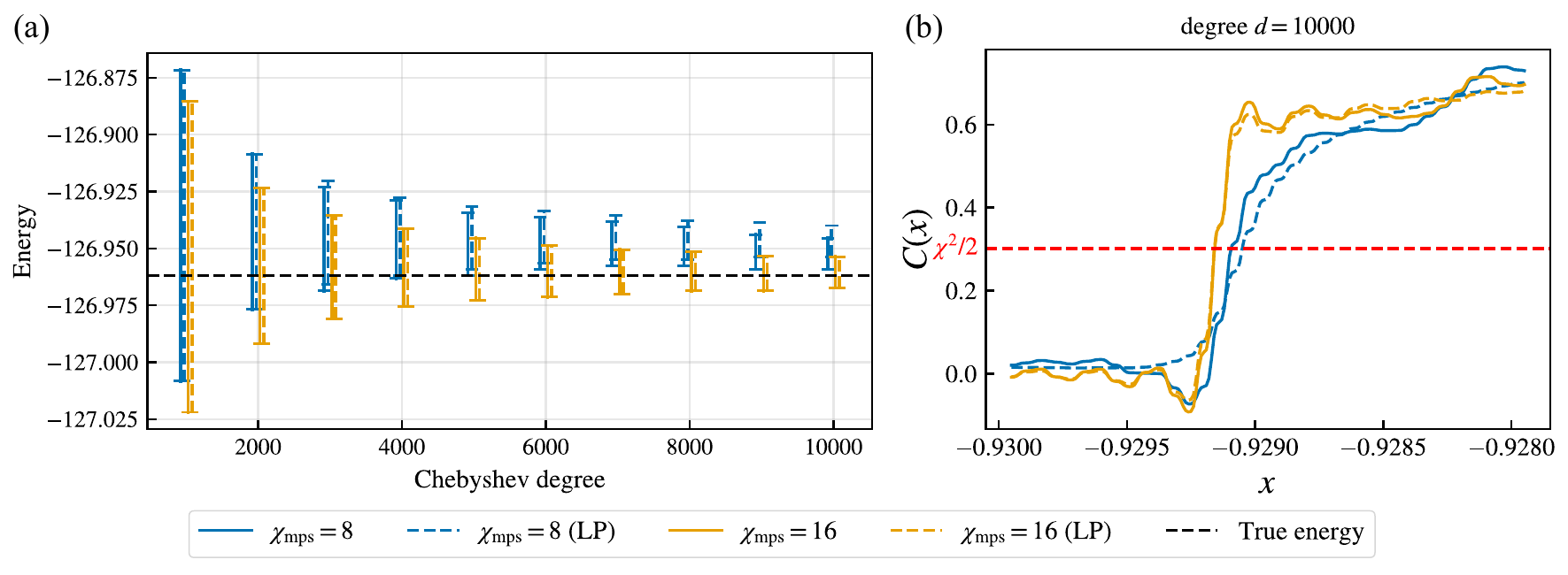}
    \caption{(a) The estimated energy is plotted as a function of Chebyshev degree for the 1D TFIM with $L=100, J=1.0, h=1.0$.
    Each color represents a different maximum bond dimension $\chi_{\mathrm{mps}}$ for the Chebyshev vectors.
    Dashed lines represent the results obtained from linear prediction (LP).
    The black dashed line represents the ground state energy obtained by DMRG with a sufficiently large bond dimension.
    The initial energy is prepared with $\chi_{\mathrm{init}}=2$.
    (b) The cumulative function $C(x)$ obtained from moments up to $10^4$ degree is plotted for different settings.
    The red horizontal line is the threshold $\chi^2/2$ to estimate the ground state energy.}
    \label{fig:fig3_1}
\end{figure*}

Fig.~\ref{fig:fig3_1} shows the results of ground state energy estimation for the 1D TFIM with $L=100, J=1.0$ and $h=1.0$, for several choices of the maximum bond dimension.
As shown in Fig.~\ref{fig:fig3_1}(a), $\chi_{\mathrm{mps}}=16$ is already sufficient to accurately approximate the Chebyshev vectors, resulting in a reliable energy estimate.
Fig.~\ref{fig:fig3_1}(b) further shows that the cumulative function begins to display the expected jump discontinuity as the bond dimension increases, approaching the ideal shape.
These observations quantitatively demonstrate that ground state energy estimation for the 1D TFIM can be easily dequantized and does not exhibit quantum advantage.

We also plot the results obtained via linear prediction as dashed lines.
In both panels (a) and (b), the extrapolated results closely match those from the direct MPS calculations.
In particular, for $\chi_{\mathrm{mps}}=16$, the linear prediction method reproduces the energy estimation and the cumulative function with near-perfect accuracy, owing to the high-quality first $1000$ moments.
These results indicate that the quantum algorithm may, in practice, be dequantized using much lower-order moments than one would expect.

\begin{figure*}[t]
    \centering
    \includegraphics[width=\linewidth]{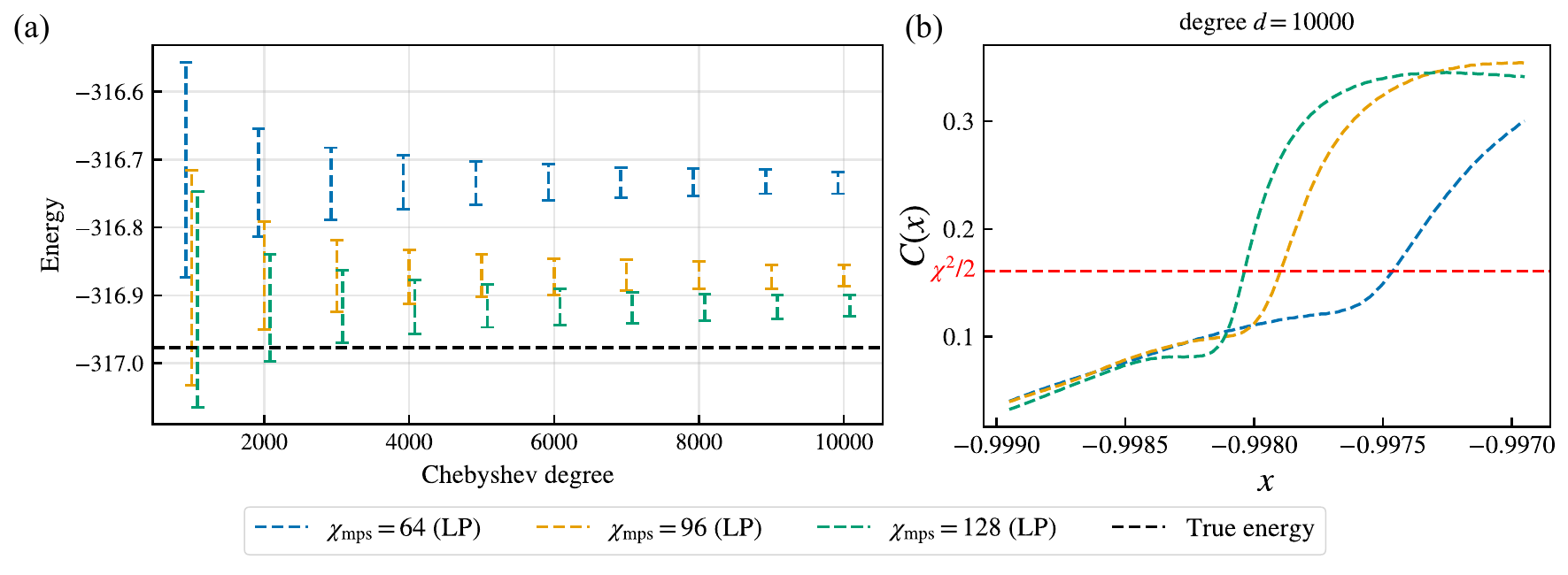}
    \caption{(a) The estimated energy is plotted as a function of Chebyshev degree for the 2D TFIM with $L=10, J=1.0, h=3.0$.
    Each color represents a different maximum bond dimension $\chi_{\mathrm{mps}}$ for the Chebyshev vectors.
    Moments beyond order $10^3$ are estimated using linear prediction (LP).
    The black dashed line represents the ground state energy obtained by DMRG with sufficient bond dimension.
    The initial energy is prepared with $\chi_{\mathrm{init}}=2$.
    (b) The cumulative function $C(x)$ obtained by moments up to $10^4$ degree is plotted for different settings.
    Red horizontal line is the threshold $\chi^2/2$ to estimate the ground state energy.}
    \label{fig:fig3_2}
\end{figure*}

Fig.~\ref{fig:fig3_2} shows the results of ground state energy estimation for the 2D TFIM on an $L=10$ lattice, corresponding to a $100$-qubit system.
Because classical computation of the Chebyshev MPS becomes increasingly expensive (each step takes several minutes for the bond dimensions considered here), we compute the moments up to order $10^3$ using MPS, and use these to extrapolate the sequence up to order $10^4$.
In contrast to the 1D case, the bond dimensions we use in the experiments are insufficient to obtain an accurate energy estimate for the 2D system due to its entanglement growth.
As a result, the quantum algorithm for the high precision regime cannot be dequantized.

Nevertheless, an important observation is that the accuracy improves systematically as the bond dimension increases.
This behavior suggests that, if the bond dimensions are increased further, the results will approach those of the ideal QSVT algorithm.
This provides a key conceptual advantage of the tensor network-based dequantization framework.
It enables an intuitive comparison between classical and quantum computational regimes through the single precision parameter $\epsilon$, which is the spirit of the dequantization.
The required precision determines the polynomial degree $d$, which directly determines the quantum computational cost.
On the classical side, the task is to accurately compute Chebyshev moments up to order $d$.
Because the approximation error decreases systematically with increasing bond dimension, this framework allows us to estimate the necessary classical computational resources to achieve a desired accuracy.

If increasing the bond dimension appears sufficient for the dequantized algorithm to succeed, then no quantum advantage exists in that regime.
Conversely, if the Chebyshev vectors and the moment sequence cannot be approximated by any classical methods while the required polynomial degree $d$ remains within the feasible range for a quantum computer, then the corresponding task may lie in a regime where quantum computation provides genuine utility.

\section{Discussion and Conclusion}\label{sec:conclusion}
In this work, we have proposed a dequantization framework for ground state energy estimation based on tensor networks.  
We first showed that the tensor network-based dequantized algorithm reproduces the computational complexity of prior dequantization results while eliminating the need for Monte Carlo sampling.
We then introduced practical variants that incorporate tensor network approximations, which run in time linear in $1/\epsilon$ when the approximation error is within a tolerable range.
Finally, we demonstrated through numerical experiments that the dequantized algorithm works for the 1D and 2D TFIM up to $100$ qubits and Chebyshev degrees of order $10^4$.
Although the 2D TFIM could not be fully dequantized with the bond dimensions used in our experiments, this limitation is itself informative: it suggests that the boundary between the classically tractable and quantum-advantageous regions.

The ultimate goal of our research is to identify the boundary between classical and quantum computational power, and to determine the regimes in which quantum computation provides a genuine advantage.
There are several possible approaches to answering this question.
One approach is to compare against the best-known classical heuristics.  
For the GSEE problem, a variety of efficient and accurate numerical algorithms have long been developed, such as DMRG.
However, because classical and quantum algorithms have fundamentally different structures, performing such comparisons in a systematic and model-independent manner is challenging. 
Another approach is the classical simulation of quantum circuits~\cite{huang_Efficient_2021, tindall_Efficient_2023}.  
This provides a unified comparison framework because all quantum algorithms are expressed as quantum circuits and can be evaluated under a single metric of classical simulation cost.
While this method has been successfully applied to NISQ algorithms, it becomes extremely challenging for FTQC algorithms involving many qubits and deep circuits.
Furthermore, such simulations generally cannot exploit problem-specific structure or incorporate approximation methods effectively.

The tensor network-based dequantization framework addresses several of these issues.  
First, the classical and quantum algorithms share the same structure, and their complexity can be compared using a single accuracy parameter~$\epsilon$, or equivalently, the polynomial degree of the filter function.
Classical simulations are limited by the degree for which the approximation remains successful.
In this sense, the hardness of the classical simulation is fully characterized by the difficulty of approximating the Chebyshev vectors.  
On the other hand, the cost of a quantum algorithm, at least in theory, scales only polynomially with the degree.
This perspective allows us to intuitively separate the regime in which classical methods remain effective from the regime in which quantum advantage is possible, using a single parameter $\epsilon$ or the Chebyshev degree $d$.

The role of approximation is also clear in this framework.
Classical executability reduces to how accurately tensor network states can represent the Chebyshev vectors and their moments.  
Different tensor network ansatz, like tree tensor network~\cite{shi_Classical_2006}, multi-scale renormalization ansatz~\cite{vidal_Entanglement_2007}, and projected entangled pair states~\cite{verstraete_Renormalization_2004}, may yield substantially better approximations depending on the structure of the underlying physical system.

An interesting direction for future work is to apply this dequantization methodology to other QSVT-based algorithms.  
The core requirement of our method is the ability to approximate the Chebyshev vectors $\ket{t_k}$ and the moments $\mu_k$.  
Thus, any QSVT-based algorithm whose output can be reconstructed solely from these quantities becomes immediately dequantizable.  
To establish quantum advantage in such settings, it is therefore necessary to show that the required polynomial degree is sufficiently large and that the Chebyshev vectors are classically hard to approximate.

\section*{Acknowledgements}
This work is supported by MEXT Quantum Leap Flagship Program (MEXT Q-LEAP) Grant No. JPMXS0120319794, JST COI-NEXT Grant No. JPMJPF2014, and JST CREST JPMJCR24I3.
T.S. is supported by MEXT Quantum Leap Flagship Program (MEXT Q-LEAP) Grant No. JPMXS0118067394, and JST PRESTO (Grant No. JPMJPR24F4).


\bibliography{reference_manabe}

\appendix

\end{document}